\documentclass[11pt]{extarticle}

\usepackage[T1]{fontenc}
\usepackage[utf8]{inputenc}
\usepackage{amsmath,amssymb,amsfonts,amsthm}
\usepackage{mathtools}
\usepackage{bm}
\usepackage{booktabs}
\usepackage{microtype}
\usepackage{graphicx}
\usepackage[top=1in,bottom=1.25in,left=0.75in,right=0.75in]{geometry}
\usepackage{hyperref}

\hypersetup{
  colorlinks=true,
  linkcolor=blue,
  citecolor=blue,
  urlcolor=blue,
  pdftitle={The diffusive spectral gap of the qubit eSWAP chain},
  pdfauthor={Yanying Liang and Haoran Zhu},
  pdfsubject={Spectral gaps of symmetric random quantum circuits},
  pdfkeywords={eSWAP circuits, unitary designs, spectral gaps, interchange process, exclusion process}
}

\newtheorem{theorem}{Theorem}[section]
\newtheorem{lemma}[theorem]{Lemma}
\newtheorem{corollary}[theorem]{Corollary}
\newtheorem{proposition}[theorem]{Proposition}

\theoremstyle{definition}

\theoremstyle{remark}

\numberwithin{equation}{section}

\newcommand{\C}{\mathbb C}
\newcommand{\R}{\mathbb R}
\newcommand{\Sn}{\mathfrak S}
\newcommand{\Univ}{\mathcal U}

\newcommand{\Specht}[1]{\mathsf S^{#1}}
\newcommand{\Schur}[1]{\mathbb S_{#1}}
\newcommand{\ip}[2]{\left\langle #1,\,#2\right\rangle}
\newcommand{\transp}[2]{(#1\,#2)}

\DeclareMathOperator{\Tr}{Tr}

\DeclareMathOperator{\Ker}{Ker}
\DeclareMathOperator{\Ran}{Ran}
\DeclareMathOperator{\Spec}{Spec}
\DeclareMathOperator{\Symm}{Sym}

\begin{document}

\title{The diffusive spectral gap of the qubit eSWAP chain}

\author{%
Yanying Liang\textsuperscript{1,2}\qquad
Haoran Zhu\textsuperscript{3}\\[0.8ex]
\small\textsuperscript{1}College of Mathematics and Informatics,
South China Agricultural University, Guangzhou 510642, China\\
\small\textsuperscript{2}Institute of Advanced Intelligence and Computing,
A*STAR, Singapore 138632, Republic of Singapore\\
\small\textsuperscript{3}Division of Mathematical Sciences,
Nanyang Technological University, Singapore 637371\\[0.5ex]
\small\href{mailto:yyl2022@scau.edu.cn}{yyl2022@scau.edu.cn}\qquad
\href{mailto:zhuh0031@e.ntu.edu.sg}{zhuh0031@e.ntu.edu.sg}}

\date{}

\maketitle

\begin{abstract}
We prove that the second-moment Hamiltonian of the open qubit eSWAP chain
has a diffusive spectral gap, resolving a conjecture of Li, Zheng, Liu,
Jiang, and Liu.  We also characterise
the ground space on every connected interaction graph with at least five
vertices and prove a uniform lower bound for the unnormalised complete-graph
gap.  We apply these results to sequential and brickwork circuits and prove
that the periodic chain also has a diffusive spectral gap.
\end{abstract}



\section{Introduction}\label{sec:introduction}

Unitary designs reproduce Haar moments up to a prescribed order and are
used in randomised benchmarking, decoupling and scrambling.  Exact and approximate
unitary $2$-designs were developed in \cite{DankertEtAl}.  Polynomial
convergence of local random circuits was proved at second moment in
\cite{HarrowLow} and at arbitrary fixed moment in
\cite{BrandaoHarrowHorodecki}; subsequent work sharpened the dependence on
moment order, depth and interaction geometry
\cite{Haferkamp,HarrowMehraban}.  The associated moment Hamiltonians connect
these questions to spectral-gap estimates and statistical mechanics
\cite{HunterJones}.  Finite-size criteria and martingale estimates give
local-to-global gap bounds \cite{Knabe,Nachtergaele,GossetMozgunov}, and
exact second-moment gaps are known for several one-dimensional ensembles
\cite{DenerisEtAl}.

For symmetric circuits, one must first identify the limiting ensemble.
Local symmetric gates may generate a proper subgroup of the globally
symmetric unitaries \cite{Marvian}; attainable design orders have since been
characterised for broad classes of symmetric circuits
\cite{MitsuhashiEtAl,LiuHulseMarvian}.  The ground space of the moment
Hamiltonian records the conserved algebra, while its low-energy spectrum
determines relaxation \cite[Sections~II and IV]{MoudgalyaMotrunich}.
Number-conserving circuits have diffusive second-moment modes
\cite{HearthEtAl,HearthErratum}.  Non-Abelian symmetric random unitaries also
arise in scrambling and covariant quantum error correction
\cite{LiEtAlApplications}.

For $n$ qubits, let $\Univ_\times$ be the group of unitaries on
$(\C^2)^{\otimes n}$ that commute with $u^{\otimes n}$ for every
$u\in\mathrm{SU}(2)$.  An $\mathrm{SU}(2)$-symmetric unitary $2$-design reproduces the second moment
of Haar measure on $\Univ_\times$.  Li, Zheng, Liu, Jiang and Liu constructed
local $\mathrm{SU}(d)$-symmetric ensembles whose design order grows
quadratically with the number of qudits \cite[Theorem~7]{LiEtAl}.
For qubits, they proved that, when $n\geqslant5$, the second-moment fixed
space of the group generated by eSWAP gates agrees with the Haar fixed
space on $\Univ_\times$
\cite[Theorem~8 and Eqs.~(C38)--(C40)]{LiEtAl}.  The remaining question is
how quickly local eSWAP circuits converge.

The eSWAP gates are $e^{-\mathrm i t\tau_e}$, where $\tau_e$ exchanges the
qubits at the endpoints of an edge $e$.  Since
$\tau_{\{i,j\}}=\tfrac12(I+X_iX_j+Y_iY_j+Z_iZ_j)$, an eSWAP gate is,
up to a global phase, isotropic Heisenberg exchange.  This interaction
preserves total spin and, with two ancillary qubits, generates every
rotationally invariant unitary \cite[Appendix~I]{MarvianLiuHulse}.
The second-moment gap measures the relaxation scale imposed by this
non-Abelian conservation law.

At each step, the circuit chooses one of the $n-1$ bonds of an open chain
uniformly and applies an eSWAP gate whose angle is sampled uniformly over a
full period.  Denote this one-step ensemble by $\mathcal E_n$.  Its second-moment
operator is
\begin{equation}\label{eq:intro-moment-operator}
 \mathsf M_n
 =\frac{1}{n-1}\sum_{i=1}^{n-1}\mathsf T_{\{i,i+1\}}^{(2)}
 =I-\frac{1}{n-1}H_n,
 \qquad
 H_n=\sum_{i=1}^{n-1}h_{\{i,i+1\}}.
\end{equation}
Here $\mathsf T_e^{(2)}$ is the local second-moment twirl and
$h_e=I-\mathsf T_e^{(2)}$.  Li et al.\ conjectured that
$\Delta(H_n)=\Theta(n^{-2})$, supported by exact finite-size spectra
\cite[Sec.~IV B, Eq.~(37), and the conjecture following Eq.~(43)]{LiEtAl}.

We prove this spectral-gap conjecture.  We show that every non-zero
eigenvalue of one half of the path Laplacian is an eigenvalue of the moment
Hamiltonian, and obtain a matching lower bound for its spectral gap.  The
lower bound follows from a uniform bound for the reaction gap on the complete
graph and an explicit solution of a Poisson equation for two-particle simple
exclusion on the path.  We then derive convergence bounds for sequential,
brickwork, periodic and all-to-all circuits.

\subsection{Main results}

For a positive semidefinite operator $L$, let $\Delta(L)$ denote its smallest
positive eigenvalue.  If $0\leqslant M\leqslant I$ and $1\in\Spec(M)$, we write
$\operatorname{gap}(M)=\Delta(I-M)$.  Define
\[
 c_0=-\frac{3\sin(\pi\sqrt2)}{10\pi\sqrt2},
 \qquad
 c_*=\left(\frac{1}{4c_0}+\frac12\right)^{-1}.
\]
Both constants are positive, with
$c_0\approx0.0650863$ and $c_*\approx0.230359$.

\begin{theorem}[The eSWAP-chain spectral-gap conjecture]\label{thm:main}
Let $n\geqslant4$.  Then
\begin{equation}\label{eq:main-refined-bounds}
 \left(
   \frac{n(n+2)}{6c_0}
   +\frac{1}{1-\cos(\pi/n)}
 \right)^{-1}
 \leqslant \Delta(H_n)
 \leqslant 1-\cos\frac{\pi}{n}.
\end{equation}
Moreover,
\begin{equation}\label{eq:path-laplacian-branch}
 \left\{1-\cos\frac{k\pi}{n}:1\leqslant k\leqslant n-1\right\}
 \subseteq\Spec(H_n).
\end{equation}
The remaining small cases satisfy $\Delta(H_2)=1$ and $\Delta(H_3)=3/8$.
Consequently, for every $n\geqslant2$,
\begin{equation}\label{eq:main-diffusive-bounds}
 \frac{c_*}{n^2}
 \leqslant \Delta(H_n)
 \leqslant 1-\cos\frac{\pi}{n}
 \leqslant \frac{\pi^2}{2n^2}.
\end{equation}
In particular, $\Delta(H_n)=\Theta(n^{-2})$.
\end{theorem}

The eigenvalues in \eqref{eq:path-laplacian-branch} give the diffusive upper
bound, while \eqref{eq:main-diffusive-bounds} supplies the matching lower
bound for $H_n$.  Uniform bond sampling contributes
the additional factor $(n-1)^{-1}$ in \eqref{eq:intro-moment-operator}.

Let $\mathsf P_n$ be the orthogonal projection onto $\Ker H_n$.
Replicas $1,2$ are the ket replicas and $3,4$ the bra replicas.  We write
$\mathcal D_n$ for the subspace fixed by all independent
simultaneous site permutations on the replica pairs $(1,3)$ and $(2,4)$, and
$\mathcal X_n$ for the corresponding fixed space associated
with $(1,4)$ and $(2,3)$.

For a graph $G$, write $H_G=\sum_{e\in E(G)}h_e$.

\begin{theorem}
\label{thm:ground-space}
Let $G$ be a connected graph on $n\geqslant5$ vertices.  Then
\begin{equation}\label{eq:intro-ground-space}
 \Ker H_G=\mathcal D_n+\mathcal X_n,
 \qquad
 \dim\Ker H_G=2\binom{n+3}{3}^2-(n+1)^4.
\end{equation}
\end{theorem}

Thus the ground space is the same for every connected interaction graph
on $n$ vertices.  Equations~\eqref{eq:direct-pairing-space} and
\eqref{eq:crossed-pairing-space} give the two fixed spaces explicitly; their
intersection is the separately symmetric space in
\eqref{eq:separately-symmetric-ground-space}.

Let $K_n$ denote the complete graph on $n$ vertices.  Its moment
Hamiltonian satisfies the following uniform bound.

\begin{theorem}
\label{thm:all-to-all}
For every $n\geqslant2$,
\begin{equation}\label{eq:complete-graph-moment-gap}
 H_{K_n}
 \geqslant c_0\left(I-\mathsf P_n\right).
\end{equation}
\end{theorem}

For uniformly sampled path edges, Theorem~\ref{thm:main} gives a
second-moment spectral gap of order $n^{-3}$.  For brickwork circuits,
the corresponding scale is $n^{-2}$ per period.  We give the precise
operator-norm and CP-order mixing statements in
Section~\ref{sec:circuits}, together with their periodic and all-to-all
counterparts.

These results complement our work on iSWAP circuits.  For every connected
graph with at least three vertices, iSWAP maximises the gap among two-local
ensembles with Hermitian second-moment operators
\cite[Theorem~1.1]{LiangZhuISWAPMaximises}.  On $K_n$, for $n\geqslant5$,
a family of ironed two-qubit gadgets attains the same gap
\cite[Theorem~1.1]{LiangZhuISWAPMatching}.  Here the conserved
$\mathrm{SU}(2)$ action instead produces a non-trivial stationary space and
a diffusive gap.

\subsection{Proof outline}

Each local Hamiltonian term is the sum of two orthogonal projections,
corresponding to interchange and reaction.  The interchange Hamiltonian acts through site permutations on
$(\C^{16})^{\otimes n}$.  Aldous' spectral-gap theorem
\cite[Theorem~1.1]{CaputoLiggettRichthammer} gives its exact path gap
$1-\cos(\pi/n)$, which controls the complement of $\Symm^n(\C^{16})$.

It remains to bound the reaction Hamiltonian on the symmetric space.
We first work on the complete graph.  The qubit fixed-space
theorem of Li et al.\ identifies its ground space with the sum of the direct
and crossed pairing spaces.  The dimensions of the two-row Specht modules
determine their principal angles.  Deleting one tensor factor from a
symmetric tensor gives a recurrence for the gap.  The coefficients in this
recurrence have a positive infinite product, so the exact three-particle
spectrum yields the uniform bound $c_0$.  The compression of the path
reaction Hamiltonian to the symmetric space then has gap at least $2c_0/n$.  The recurrence is related to
Caputo's complete-graph method for binary-collision generators
\cite{CaputoKac} and Sasada's complete-graph-to-lattice comparisons
\cite{Sasada}.

To combine these bounds, we control the off-diagonal reaction block between
the symmetric space and its orthogonal complement.  We use the permutation
representation on unordered pairs for this purpose.  The centred indicator
of the path edges lies in the
two non-trivial irreducible summands, $\Specht{(n-1,1)}$ and
$\Specht{(n-2,2)}$.  We solve a Poisson equation for this vector under the
two-particle simple exclusion process.  The Rayleigh quotient of the
Moore--Penrose inverse, evaluated at either component, is $(n-1)/3$, whereas
the norm of that inverse has order $n^2$.  Combining these identities with the complete-graph
bound gives the required two-scale Poincar\'e inequality.

The paper is organised as follows.  In Section~\ref{sec:local-structure}, we
derive the local interchange--reaction decomposition.  In
Section~\ref{sec:transport}, we obtain eigenvalues from the graph Laplacian and prove
the interchange gap.  In Section~\ref{sec:ground-geometry}, we determine the
geometry of the two global pairing spaces.  In Section~\ref{sec:complete-gap},
we prove the uniform complete-graph reaction gap.  In
Section~\ref{sec:edge-defect}, we solve the Poisson equation on unordered pairs; in
Section~\ref{sec:slow-fast}, we use it to control the off-diagonal block.  In
Section~\ref{sec:two-scale}, we complete the lower bound.  In Section~\ref{sec:circuits}, we derive
the circuit convergence estimates.  In Section~\ref{sec:discussion}, we
discuss further spectral questions.
In Appendix~\ref{app:three-site}, we give the exact three-site calculations.

\section{Local structure of the eSWAP moment Hamiltonian}\label{sec:local-structure}

We compute the local second-moment twirl and decompose
its complementary projection into interchange and reaction projections.

All Hilbert spaces in this paper are finite-dimensional and complex, and
all operator norms are induced by the Hilbert-space norm.  We write
$\Sn_n$ for the symmetric group on $n$ letters.  If $X\geqslant0$, then $X^+$ denotes its Moore--Penrose inverse
and $X^{+1/2}=(X^+)^{1/2}$.

Let $\mathcal K_n=(\C^2)^{\otimes n}$ be the Hilbert space of
$n$ qubits.  Its second-moment space is
$\mathcal H_n=\mathcal K_n^{\otimes2}\otimes\overline{\mathcal K_n}^{\otimes2}$.
The first two factors are the ket replicas and the last two are the bra
replicas.  Every qubit transposition is real in the computational basis, so
we identify the conjugate factors with copies of $\mathcal K_n$.

For an unordered pair $e=\{i,j\}\subseteq\{1,\ldots,n\}$, let $\tau_e$ be the
operator that exchanges the $i$th and $j$th physical qubits.  Thus $\tau_e$
is self-adjoint and $\tau_e^2=I$.  Write $s_e^{(a)}$ for the operator that
acts as $\tau_e$ on replica $a$ and as the identity on the other replicas.

The local second-moment twirl associated with $e$ is the orthogonal projection
\begin{equation}\label{eq:local-second-moment-twirl}
 \mathsf T_e^{(2)}
 =\frac{1}{2\pi}\int_0^{2\pi}
 \bigl(e^{-\mathrm i t\tau_e}\bigr)^{\otimes2}
 \otimes
 \bigl(e^{\mathrm i t\tau_e}\bigr)^{\otimes2}\,dt.
\end{equation}
Let $h_e=I-\mathsf T_e^{(2)}$ be the local Hamiltonian term, and write
$Z_e
=s_e^{(1)}+s_e^{(2)}-s_e^{(3)}-s_e^{(4)}$.
The four transposition operators commute and have eigenvalues in $\{1,-1\}$,
so $\Spec(Z_e)\subseteq\{-4,-2,0,2,4\}$.

The integrand in \eqref{eq:local-second-moment-twirl} is
$e^{-\mathrm i tZ_e}$.  Integration annihilates every non-zero eigenspace
of $Z_e$ and fixes its kernel.  Thus
\begin{equation}\label{eq:twirl-as-zero-projection}
 \mathsf T_e^{(2)}=\chi_{\{0\}}(Z_e),
 \qquad
 h_e=\chi_{\R\setminus\{0\}}(Z_e).
\end{equation}
Equivalently, $h_e=\tfrac{5}{16}Z_e^2-\tfrac{1}{64}Z_e^4$, since the
polynomial $5x^2/16-x^4/64$ equals zero at $0$ and one at $\pm2,\pm4$.

To decompose $h_e$, we group the four replicas at each site.  The one-site space is
$V=(\C^2)^{\otimes4}\cong\C^{16}$, so
$\mathcal H_n\cong V^{\otimes n}$.  The operator
$S_e
=s_e^{(1)}s_e^{(2)}s_e^{(3)}s_e^{(4)}$
exchanges the two sites at the endpoints of $e$.  Define
$p_e
=\tfrac12(I-S_e)$.
This is the orthogonal projection onto the subspace that is antisymmetric
under exchange of the two sites.

For $a\in\{1,2,3,4\}$, set
$\mathsf P_{a,\pm}^{(e)}=\tfrac12(I\pm s_e^{(a)})$.
The local reaction projection is
\begin{equation}\label{eq:local-reaction-projection}
 q_e
 =\mathsf P_{1,+}^{(e)}\mathsf P_{2,+}^{(e)}
  \mathsf P_{3,-}^{(e)}\mathsf P_{4,-}^{(e)}
 +\mathsf P_{1,-}^{(e)}\mathsf P_{2,-}^{(e)}
  \mathsf P_{3,+}^{(e)}\mathsf P_{4,+}^{(e)}.
\end{equation}

\begin{proposition}\label{prop:local-interchange-reaction-decomposition}
For every edge $e$,
\begin{equation}\label{eq:local-interchange-reaction-decomposition}
 h_e
 =p_e
  +q_e,
 \qquad
 p_eq_e=0.
\end{equation}
\end{proposition}

\begin{proof}
Simultaneously diagonalise $s_e^{(1)},\ldots,s_e^{(4)}$, and let
$x_1,\ldots,x_4\in\{1,-1\}$ be their joint eigenvalues.  If
$x_1x_2x_3x_4=-1$, then
$x_1+x_2-x_3-x_4\in\{2,-2\}$.
Hence $h_e=1$, $p_e=1$, and $q_e=0$ on that joint eigenspace.

Suppose that $x_1x_2x_3x_4=1$.  Then
$p_e=0$, and the
difference $x_1+x_2-x_3-x_4$ is non-zero precisely for
$(x_1,x_2,x_3,x_4)=(1,1,-1,-1)$ or $(-1,-1,1,1)$.
The range of the projection in \eqref{eq:local-reaction-projection} is the
sum of these two joint eigenspaces.  Hence
\eqref{eq:local-interchange-reaction-decomposition} holds on every joint
eigenspace.  The two summands have disjoint supports and are therefore
orthogonal.
\end{proof}

For a graph $G$ on $\{1,\ldots,n\}$, let
$E_G=\sum_{e\in E(G)}p_e$ and $R_G=\sum_{e\in E(G)}q_e$.
Proposition~\ref{prop:local-interchange-reaction-decomposition} gives
\begin{equation}\label{eq:graph-interchange-reaction-decomposition}
 H_G
 =E_G+R_G.
\end{equation}
Let $P_n$ denote the path with edges $\{i,i+1\}$,
$1\leqslant i\leqslant n-1$, and let $K_n$ denote the complete graph.  The
Hamiltonian studied in this paper is $H_n=H_{P_n}$.
Every local summand is an orthogonal projection, and $H_n$ is frustration-free.

The local kernel has an equivalent pairing description.  Define
\[
 \mathsf D_e
 =\tfrac14(I+s_e^{(1)}s_e^{(3)})(I+s_e^{(2)}s_e^{(4)}),
 \qquad
 \mathsf X_e
 =\tfrac14(I+s_e^{(1)}s_e^{(4)})(I+s_e^{(2)}s_e^{(3)}).
\]
These projections correspond to the direct pairing $(1,3)(2,4)$ and the crossed pairing
$(1,4)(2,3)$.  Simultaneous diagonalisation gives
$\mathsf T_e^{(2)}=\mathsf D_e+\mathsf X_e-\mathsf D_e\mathsf X_e$ and
$h_e=(I-\mathsf D_e)(I-\mathsf X_e)$.  In Section~\ref{sec:ground-geometry},
we use the global versions of these pairing spaces to determine the ground
space of $H_n$.

\section{Laplacian eigenvalues and the interchange gap}\label{sec:transport}

We obtain the upper bound by realising one half of the graph Laplacian on
a subspace of the moment space.  We then use Aldous' theorem to determine
the gap of the interchange Hamiltonian.

For a partition $\lambda$ of $n$, let $\Specht{\lambda}$ be the
corresponding Specht module and let $\Schur{\lambda}W$ be the Schur module
of a complex vector space $W$.  Write $\ell(\lambda)$ for the number of
non-zero parts of $\lambda$.  Schur--Weyl duality
\cite[Section~6.1]{FultonHarris} gives
\begin{equation}\label{eq:qubit-schur-weyl}
 (\C^2)^{\otimes n}
 \cong
 \bigoplus_{\substack{\lambda\vdash n\\\ell(\lambda)\leqslant2}}
 \Schur{\lambda}\C^2\otimes\Specht{\lambda}.
\end{equation}
Both $\Specht{(n)}$ and $\Specht{(n-1,1)}$ occur for every $n\geqslant2$.
We realise the standard representation on the zero-sum subspace
$\mathcal V_n=\mathbf 1^\perp\subseteq\C^n$ and write $\rho$ for its
$\Sn_n$-action.

Consider a tensor product with one standard Specht module and three
trivial Specht modules.  For definiteness, take
$\Specht{(n)}\otimes\Specht{(n)}\otimes\Specht{(n-1,1)}\otimes\Specht{(n)}$.
The transpositions in the first, second, and fourth replicas act trivially
on this tensor product, while those in the third act through the standard
representation.  Hence
$Z_e=I-\rho(\transp{i}{j})$ for every edge
$e=\{i,j\}$.
Since $\rho(\transp{i}{j})$ is an involution, \eqref{eq:twirl-as-zero-projection} gives
\begin{equation}\label{eq:local-moment-standard-sector}
 h_e
 =\frac12\bigl(I-\rho(\transp{i}{j})\bigr)
\end{equation}
on this tensor product.

For a graph $G$, let $L_G$ be its combinatorial Laplacian, with
$(L_G\mathbf x)_i=\sum_{j:\{i,j\}\in E(G)}(x_i-x_j)$.

\begin{proposition}\label{prop:diffusive-upper-bound}
Let $G$ be a connected graph on $n\geqslant2$ vertices.  Then
\begin{equation}\label{eq:general-laplacian-branch}
 \Spec\!\left(\left.\frac12L_G\right|_{\mathbf 1^\perp}\right)
 \subseteq\Spec\bigl(H_G\bigr).
\end{equation}
More precisely, on every copy of each of the four tensor products obtained by
placing $\Specht{(n-1,1)}$ in one replica and $\Specht{(n)}$ in the other
three, $H_G$ is unitarily equivalent to
$\tfrac12L_G|_{\mathbf 1^\perp}$.  Thus each eigenvalue of
$\tfrac12L_G|_{\mathbf 1^\perp}$ of multiplicity $m$ occurs in these four
types of summands with total multiplicity
$4(n-1)(n+1)^3m$.  Consequently,
$\Delta(H_G)\leqslant\lambda_2(L_G)/2$.
For the path, this gives \eqref{eq:path-laplacian-branch}, with eigenvectors
in the standard factor having coordinates
\begin{equation}\label{eq:path-laplacian-eigenvectors}
 x_j^{(k)}=\sqrt{\frac2n}
 \cos\!\left(\frac{k\pi}{n}\left(j-\frac12\right)\right),
 \qquad 1\leqslant j\leqslant n,\quad 1\leqslant k\leqslant n-1.
\end{equation}
In particular, $\Delta(H_n)\leqslant1-\cos(\pi/n)$.
\end{proposition}

\begin{proof}
For $\mathbf x\in\mathcal V_n$, equation
\eqref{eq:local-moment-standard-sector} gives
\[
 \ip{\mathbf x}{H_G\mathbf x}
 =\frac12\sum_{\{i,j\}\in E(G)}|x_i-x_j|^2
 =\frac12\ip{\mathbf x}{L_G\mathbf x}.
\]
Thus the restriction of $H_G$ to this tensor product is $\tfrac12L_G$ on $\mathbf 1^\perp$, which proves
\eqref{eq:general-laplacian-branch}.  Placing the standard factor in another
replica changes $Z_e$ at most by an overall sign, which leaves
its non-zero spectral projection unchanged; hence the same restriction occurs
in all four blocks.  Since
$\dim\Schur{(n)}\C^2=n+1$ and
$\dim\Schur{(n-1,1)}\C^2=n-1$, the multiplicity assertion follows.  For the
path, $L_{P_n}$ has eigenvalues
$2(1-\cos(k\pi/n))$, where $k=0,1,\ldots,n-1$.
The $k=0$ eigenvector is constant and is absent from
$\mathcal V_n$.  The discrete-cosine eigenvectors for the
remaining eigenvalues are given by \eqref{eq:path-laplacian-eigenvectors}.
This proves \eqref{eq:path-laplacian-branch}, whose smallest eigenvalue gives
the asserted gap bound.
\end{proof}

Since $1-\cos t\leqslant t^2/2$, this also gives
$\Delta(H_n)\leqslant\pi^2/(2n^2)$.

We next determine the gap of the interchange Hamiltonian.  Its kernel is the symmetric subspace
$\mathcal S_n=\Symm^n(V)\subseteq V^{\otimes n}$.
Its orthogonal projection is
$\mathsf Q_n=\frac{1}{n!}\sum_{\sigma\in\Sn_n}U_\sigma$,
where $U_\sigma$ permutes the $n$ sites.

\begin{proposition}\label{prop:interchange-gap}
For every $n\geqslant2$,
\begin{equation}\label{eq:interchange-exact-coercivity}
 E_{P_n}
 \geqslant
 \left(1-\cos\frac{\pi}{n}\right)
 \left(I-\mathsf Q_n\right).
\end{equation}
Moreover, $\Delta(E_{P_n})=1-\cos(\pi/n)$.  In particular,
$E_{P_n}\geqslant(2/n^2)(I-\mathsf Q_n)$.
\end{proposition}

\begin{proof}
The operator $2E_{P_n}=\sum_{i=1}^{n-1}(I-S_{\{i,i+1\}})$ is the positive generator of the interchange process on the path, acting in
the permutation representation on $V^{\otimes n}$.  Aldous' spectral-gap
theorem states that the gap of the interchange process on any weighted graph
equals that of the corresponding one-particle random walk
\cite[Theorem~1.1]{CaputoLiggettRichthammer}.  Every irreducible representation
of $\Sn_n$ occurs in the regular representation, so this theorem bounds the
positive spectrum in each non-trivial irreducible representation from below
by the one-particle gap.  It therefore holds on $V^{\otimes n}$.  Since the path Laplacian has gap
$2(1-\cos(\pi/n))$, we obtain \eqref{eq:interchange-exact-coercivity}.
The kernel is the common fixed space of the adjacent simultaneous exchanges in all four replicas,
namely $\mathcal S_n$.

To see that the bound is attained, choose orthonormal vectors
$\mathbf a,\mathbf b\in V$ and let
$\xi_i=\mathbf a^{\otimes(i-1)}\otimes\mathbf b\otimes\mathbf a^{\otimes(n-i)}$.
The span of $\xi_1,\ldots,\xi_n$ is the natural permutation module,
whose zero-sum subspace is $\Specht{(n-1,1)}$.  Hence the path gap is attained.

Finally, $1-\cos(\pi/n)=2\sin^2(\pi/(2n))\geqslant2/n^2$,
where we used $\sin t\geqslant2t/\pi$ for $0\leqslant t\leqslant\pi/2$.  This gives the final inequality.
\end{proof}

We now identify the common ground space of connected graphs.  The operator
$R_{K_n}$ commutes with all site permutations and hence preserves
$\mathcal S_n$.  Let $R_n$ denote its restriction to this subspace.
Write $\mathcal G_n=\Ker H_n$, so that $\mathsf P_n$ is the
orthogonal projection onto $\mathcal G_n$.

\begin{proposition}\label{prop:path-complete-ground-space-equality}
Let $G$ be a connected graph on $n\geqslant2$ vertices.  Then
\begin{equation}\label{eq:path-complete-ground-space-equality}
 \Ker H_G=\mathcal G_n=\Ker R_n
 \subseteq\mathcal S_n.
\end{equation}
\end{proposition}

\begin{proof}
Let $\psi\in\Ker H_G$.  Positivity of the local summands
implies that every interchange and reaction projection on an edge of $G$
annihilates $\psi$.  The edge transpositions of a connected graph generate
$\Sn_n$, so $\psi\in\mathcal S_n$.
Choose an edge $e_0$ of $G$.  For any unordered pair $e$, take a site
permutation $\sigma$ sending $e_0$ to $e$.  Then
$q_e\psi=U_\sigma q_{e_0}U_\sigma^{-1}\psi=0$, so $\psi\in\Ker R_n$.  Conversely, a vector in $\Ker R_n$ is
completely symmetric and, by positivity, is annihilated by every reaction
projection.  It therefore belongs to $\Ker H_G$.
Taking $G=P_n$ identifies this common kernel with $\mathcal G_n$.
\end{proof}

\section{The ground space and the two pairings}\label{sec:ground-geometry}

We now describe the ground space as the sum of two global pairing spaces
and determine their intersection and principal angles.

Let $W_1,W_2,W_3,W_4\cong\C^2$ be the four one-qubit replica spaces, so that
$V=W_1\otimes W_2\otimes W_3\otimes W_4$.  Write $\rho_a(\sigma)$ for the
action of $\sigma\in\Sn_n$ on the $n$ tensor factors of replica $a$.

Define the direct and crossed pairing spaces by
\begin{align}
 \mathcal D_n
 &=\left\{
  \xi:
  \rho_1(\sigma)\rho_3(\sigma)\xi=\xi,
  \quad
  \rho_2(\tau)\rho_4(\tau)\xi=\xi
  \text{ for all }\sigma,\tau\in\Sn_n
 \right\},\label{eq:direct-pairing-space}\\
 \mathcal X_n
 &=\left\{
  \xi:
  \rho_1(\sigma)\rho_4(\sigma)\xi=\xi,
  \quad
  \rho_2(\tau)\rho_3(\tau)\xi=\xi
  \text{ for all }\sigma,\tau\in\Sn_n
 \right\}.\label{eq:crossed-pairing-space}
\end{align}
Both spaces lie in $\mathcal S_n$: invariance under the two
independent paired actions implies invariance under simultaneous permutation
of all four replicas.  Regrouping the paired replica factors gives the
canonical identifications
\begin{align}
 \mathcal D_n
 &\cong
 \Symm^n(W_1\otimes W_3)
 \otimes
 \Symm^n(W_2\otimes W_4),
 \label{eq:direct-pairing-identification}\\
 \mathcal X_n
 &\cong
 \Symm^n(W_1\otimes W_4)
 \otimes
 \Symm^n(W_2\otimes W_3).
 \label{eq:crossed-pairing-identification}
\end{align}
Let $\mathsf D_n$ and $\mathsf X_n$ denote the corresponding orthogonal projections.

\begin{proposition}\label{prop:complete-graph-ground-space}
For $n\geqslant5$,
\begin{equation}\label{eq:complete-graph-ground-space}
 \mathcal G_n
 =\mathcal D_n+\mathcal X_n.
\end{equation}
\end{proposition}

\begin{proof}
We first show that both pairing spaces lie in $\mathcal G_n$.  If
$\xi\in\mathcal D_n$, then $s_e^{(1)}s_e^{(3)}\xi=\xi$ and
$s_e^{(2)}s_e^{(4)}\xi=\xi$ for every pair $e$.  Neither of the joint sign patterns supporting $q_e$
in \eqref{eq:local-reaction-projection} satisfies these identities.  Hence
$q_e\xi=0$, and
Proposition~\ref{prop:path-complete-ground-space-equality} gives
$\xi\in\mathcal G_n$.  The same argument applies to the crossed pairing
space, with the third and fourth replicas interchanged.

For the reverse inclusion, $S_e=I$ on $\mathcal S_n$, so
Proposition~\ref{prop:local-interchange-reaction-decomposition} gives
$h_e\xi=q_e\xi$ for $\xi\in\mathcal S_n$.  Consequently, $\mathcal G_n$ is the common fixed space of the second-moment twirls associated with all eSWAP one-parameter subgroups.

Let $\mathcal P_2(n)$ be the set of partitions of $n$ with at most two rows.
In the Schur--Weyl decomposition \eqref{eq:qubit-schur-weyl}, the group
$\Univ_\times$ consists of the operators
\[
 \bigoplus_{\lambda\in\mathcal P_2(n)}
 I_{\Schur{\lambda}\C^2}\otimes U_\lambda,
 \qquad U_\lambda\in\mathrm U(\Specht{\lambda}).
\]
By the results of Li et al., for $n\geqslant5$ the compact group generated by
qubit eSWAP gates has the same second-moment fixed space as $\Univ_\times$
\cite[Theorem~8 and Eqs.~(C38)--(C40)]{LiEtAl}.  It remains to compute the
$\Univ_\times$ invariants.

The group $\Univ_\times$ acts as the product of the full unitary groups on
the Specht factors and trivially on the Schur factors.  Thus a four-replica
invariant can occur only when the partitions indexing the two ket factors
coincide, up to order, with those indexing the two bra factors.  Schur--Weyl duality on the two ket and two bra
factors shows that the invariant tensors are spanned by the direct
contraction $(1,3)(2,4)$ and the crossed contraction $(1,4)(2,3)$.  When a
Specht factor is one-dimensional, the two contractions on that factor coincide; this is
accounted for by the intersection of their spans.  These spans are precisely
\eqref{eq:direct-pairing-space} and \eqref{eq:crossed-pairing-space}, proving
\eqref{eq:complete-graph-ground-space}.
\end{proof}

To compare the two pairing spaces, we compute their principal angles.
The principal-angle cosines of two subspaces are the singular values of
the orthogonal projection from one subspace to the other; see
\cite[Theorem~2]{Halmos}.  For $\lambda\vdash n$, write
$f^\lambda=\dim\Specht{\lambda}$.
Realise $\Specht{\lambda}$ by real orthogonal matrices and choose an
orthonormal basis $\mathbf e_1,\ldots,\mathbf e_{f^\lambda}$ in its real form.
Define the normalised invariant vector
\[
 \Omega_\lambda^{ab}
 =\frac{1}{\sqrt{f^\lambda}}
   \sum_{j=1}^{f^\lambda}
   \mathbf e_j^{(a)}\otimes\mathbf e_j^{(b)}.
\]

The Schur--Weyl decompositions of the two pairing spaces are
\begin{align}
 \mathcal D_n
 &=\bigoplus_{\lambda,\mu\in\mathcal P_2(n)}
 \left(
  \Schur{\lambda}W_1\otimes\Schur{\mu}W_2
  \otimes\Schur{\lambda}W_3\otimes\Schur{\mu}W_4
 \right)\otimes
 \C\bigl(\Omega_\lambda^{13}\otimes\Omega_\mu^{24}\bigr),
 \label{eq:direct-pairing-schur-weyl}\\
 \mathcal X_n
 &=\bigoplus_{\lambda,\mu\in\mathcal P_2(n)}
 \left(
  \Schur{\lambda}W_1\otimes\Schur{\mu}W_2
  \otimes\Schur{\mu}W_3\otimes\Schur{\lambda}W_4
 \right)\otimes
 \C\bigl(\Omega_\lambda^{14}\otimes\Omega_\mu^{23}\bigr).
 \label{eq:crossed-pairing-schur-weyl}
\end{align}
A tensor product of four Schur modules occurs in both decompositions only
when all four modules are indexed by the same partition.

\begin{lemma}\label{lem:pairing-vector-overlap}
For every $\lambda\in\mathcal P_2(n)$,
\begin{equation}\label{eq:pairing-vector-overlap}
 \ip{\Omega_\lambda^{13}\otimes\Omega_\lambda^{24}}
     {\Omega_\lambda^{14}\otimes\Omega_\lambda^{23}}
 =\frac{1}{f^\lambda}.
\end{equation}
\end{lemma}

\begin{proof}
Expanding the two normalised vectors gives
\begin{align*}
 \Omega_\lambda^{13}\otimes\Omega_\lambda^{24}
 =\frac{1}{f^\lambda}\sum_{a,b}|a,b,a,b\rangle,\quad\text{and}\quad
 \Omega_\lambda^{14}\otimes\Omega_\lambda^{23}
 =\frac{1}{f^\lambda}\sum_{c,d}|c,d,d,c\rangle.
\end{align*}
A summand contributes to the inner product only when $a=b=c=d$.  There are $f^\lambda$ such summands, proving \eqref{eq:pairing-vector-overlap}.
\end{proof}

Define the separately symmetric space by
\begin{equation}\label{eq:separately-symmetric-ground-space}
 \mathcal F_n
 =\Symm^n(W_1)\otimes\Symm^n(W_2)
  \otimes\Symm^n(W_3)\otimes\Symm^n(W_4).
\end{equation}
This space lies in both pairing spaces.  Let $\mathsf F_n$ denote its
orthogonal projection.

The overlap formula reduces the angle estimate to the following dimension
bound.

\begin{lemma}\label{lem:minimum-two-row-specht-dimension}
Let $n\geqslant5$ and let $\lambda\neq(n)$ be a partition of $n$ with at most
two rows.  Then $f^\lambda\geqslant n-1$.
\end{lemma}

\begin{proof}
Write $\lambda=(n-r,r)$ with $1\leqslant r\leqslant\lfloor n/2\rfloor$.  The hook-length formula \cite[Theorem~3.10.2]{Sagan} gives
\begin{equation}\label{eq:two-row-specht-dimension}
 f^{(n-r,r)}
 =\binom nr-\binom n{r-1}
 =\frac{n-2r+1}{n-r+1}\binom nr.
\end{equation}
For $r=1$, this equals $n-1$.  Suppose $r\geqslant2$ and $r\leqslant(n-1)/2$.  Then $n-2r+1\geqslant2$, $n-r+1\leqslant n-1$, and $\binom nr\geqslant\binom n2$, so \eqref{eq:two-row-specht-dimension} is at least $n$.

When $n=2m$ and $r=m$, the dimension $f^{(m,m)}$ is the Catalan number
$C_m=\frac{1}{m+1}\binom{2m}{m}$.
We have $C_3=5$, and $C_{m+1}/C_m=2(2m+1)/(m+2)>2$ for $m\geqslant3$.
It follows by induction that $C_m\geqslant2m-1$.
\end{proof}

\begin{proposition}\label{prop:pairing-principal-angle}
For $n\geqslant5$,
\begin{equation}\label{eq:pairing-principal-angle}
 \left\|
  \mathsf D_n\mathsf X_n-\mathsf F_n
 \right\|
 \leqslant\frac{1}{n-1}.
\end{equation}
In particular, $\mathcal D_n\cap\mathcal X_n=\mathcal F_n$ and
\begin{equation}\label{eq:ground-projection-pairing-bound}
 \mathsf P_n
 \leqslant\frac{n-1}{n-2}
 \left(\mathsf D_n+\mathsf X_n\right).
\end{equation}
\end{proposition}

\begin{proof}
By \eqref{eq:direct-pairing-schur-weyl} and
\eqref{eq:crossed-pairing-schur-weyl}, only the summands whose four Schur modules
are indexed by the same partition $\lambda$ contribute to the product of
the pairing projections.  On each such summand,
Lemma~\ref{lem:pairing-vector-overlap} gives overlap $1/f^\lambda$.
Removing $\mathcal F_n$ excludes $\lambda=(n)$, so
Lemma~\ref{lem:minimum-two-row-specht-dimension} gives
\eqref{eq:pairing-principal-angle}.  The bound is strictly smaller than $1$;
hence $\mathcal D_n\cap\mathcal X_n=\mathcal F_n$.

The two-projection decomposition \cite[Theorem~2]{Halmos} implies that if subspaces
$\mathcal M$ and $\mathcal N$ have zero intersection and
$\|\mathsf P_{\mathcal M}\mathsf P_{\mathcal N}\|\leqslant s<1$, then
\[
 \mathsf P_{\mathcal M+\mathcal N}
 \leqslant\frac{1}{1-s}
 \left(\mathsf P_{\mathcal M}+\mathsf P_{\mathcal N}\right).
\]
Apply this to $\mathcal D_n\ominus\mathcal F_n$ and
$\mathcal X_n\ominus\mathcal F_n$ with $s=1/(n-1)$.  Since these spaces
are orthogonal to $\mathcal F_n$,
\begin{align*}
 \mathsf P_n
 &\leqslant\mathsf F_n
 +\frac{n-1}{n-2}
 \left[
  \left(\mathsf D_n-\mathsf F_n\right)
  +\left(\mathsf X_n-\mathsf F_n\right)
 \right]\leqslant\frac{n-1}{n-2}
 \left(\mathsf D_n+\mathsf X_n\right),
\end{align*}
which is \eqref{eq:ground-projection-pairing-bound}.
\end{proof}

The identifications \eqref{eq:direct-pairing-identification} and
\eqref{eq:crossed-pairing-identification} give
$\dim\mathcal D_n=\dim\mathcal X_n
=\binom{n+3}{3}^2$, while
$\dim\mathcal F_n=(n+1)^4$.
Combining Propositions~\ref{prop:complete-graph-ground-space} and
\ref{prop:pairing-principal-angle} with the dimension formula for a sum of
subspaces, we obtain the following result.

\begin{corollary}\label{cor:ground-space-dimension}
Let $n\geqslant5$.  Then
$\dim\mathcal G_n=2\binom{n+3}{3}^2-(n+1)^4$.
\end{corollary}

Combining this corollary with
Propositions~\ref{prop:path-complete-ground-space-equality} and
\ref{prop:complete-graph-ground-space} proves
Theorem~\ref{thm:ground-space}.

\section{The reaction gap on the complete graph}\label{sec:complete-gap}

Write $\delta_n=\Delta(R_n)$, with spectra and eigenvalue multiplicities
taken on $\mathcal S_n$.  For operator inequalities on $\mathcal H_n$, we
extend $R_n$ by zero on $\mathcal S_n^\perp$.  We prove the following uniform bound.

\begin{theorem}\label{thm:uniform-complete-graph-gap}
For every $n\geqslant2$,
\begin{equation}\label{eq:uniform-complete-graph-gap}
 R_n
 \geqslant c_0
 \left(\mathsf Q_n-\mathsf P_n\right).
\end{equation}
\end{theorem}

We first use the angle between the two pairing spaces to obtain a recurrence
for the gap.  The three-particle spectrum and a finite-size inequality give
the initial bound at $n=5$.  We then show that the product of the recurrence
coefficients is bounded away from zero.

\subsection{Deletion estimates}

Fix an orthonormal basis $\mathbf v_1,\ldots,\mathbf v_{16}$ of $V$, and let
$a_1,\ldots,a_{16}$ be the corresponding annihilation operators on bosonic
Fock space, normalised by
$[a_\alpha,a_\beta^*]=\delta_{\alpha\beta}I$ and $a_\alpha|0\rangle=0$.
On the $n$-particle sector $\mathcal S_n=\Symm^n(V)$, we have
$\sum_{\alpha=1}^{16}a_\alpha^*a_\alpha=nI$.  Define the canonical one-particle deletion map by
\begin{equation}\label{eq:canonical-deletion-isometry}
 \iota_n:\mathcal S_n
 \longrightarrow V\otimes\mathcal S_{n-1},
 \qquad
 \iota_n\xi
 =\frac{1}{\sqrt n}
  \sum_{\alpha=1}^{16}\mathbf v_\alpha\otimes a_\alpha\xi.
\end{equation}
The number-operator identity shows that $\iota_n$ is an isometry.

Set $V_{13}=W_1\otimes W_3$ and $V_{24}=W_2\otimes W_4$, so that
$V=V_{13}\otimes V_{24}$.  Define the direct and crossed deletion operators by
\[
 M_n=n\,\iota_n^*(I_V\otimes\mathsf D_{n-1})\iota_n,
 \qquad
 N_n=n\,\iota_n^*(I_V\otimes\mathsf X_{n-1})\iota_n.
\]
Thus $M_n=\sum_{\alpha=1}^{16}
a_\alpha^*\mathsf D_{n-1}a_\alpha$, and similarly for
$N_n$.

\begin{proposition}\label{prop:single-pairing-deletion-spectrum}
Let $n\geqslant3$.  Then the non-zero eigenvalues of $M_n$ are
$n$ on $\mathcal D_n$ and $n/(n-1)$ on the subspace
\begin{equation}\label{eq:hook-cauchy-component}
 \Schur{(n-1,1)}V_{13}
 \otimes
 \Schur{(n-1,1)}V_{24}.
\end{equation}
In particular, $M_n\leqslant\frac{n}{n-1}I$ on
$\mathcal S_n\ominus\mathcal D_n$.
The analogous statement holds for $N_n$.
\end{proposition}

\begin{proof}
Since $\dim V_{13}=\dim V_{24}=4$, the Cauchy decomposition \cite[Exercise~6.11(b)]{FultonHarris} gives
\begin{equation}\label{eq:cauchy-decomposition}
 \Symm^n(V_{13}\otimes V_{24})
 \cong
 \bigoplus_{\substack{\lambda\vdash n\\\ell(\lambda)\leqslant4}}
 \Schur{\lambda}V_{13}\otimes\Schur{\lambda}V_{24}.
\end{equation}
Here $\mathcal D_n$ is the summand corresponding to $(n)$.
The operator $M_n$ commutes with
$\mathrm U(V_{13})\times\mathrm U(V_{24})$.  Since the decomposition is
multiplicity-free, Schur's lemma \cite[Lemma~1.7]{FultonHarris} implies that
$M_n$ is scalar on each summand.

If $M_n$ is non-zero on the $\lambda$-summand, then
$\Schur{\lambda}V_{13}$ must occur in
$V_{13}\otimes\Schur{(n-1)}V_{13}$, and likewise for $V_{24}$.
By the Pieri rule \cite[Eq.~(6.8)]{FultonHarris},
\[
 V_{13}\otimes\Schur{(n-1)}V_{13}
 \cong
 \Schur{(n)}V_{13}\oplus\Schur{(n-1,1)}V_{13}.
\]
Hence only $(n)$ and $(n-1,1)$ can contribute.

If $\xi\in\mathcal D_n$, then $a_\alpha\xi\in\mathcal D_{n-1}$ for
every $\alpha$, so $M_n\xi=\sum_{\alpha=1}^{16}a_\alpha^*a_\alpha\xi=n\xi$.
It remains to determine the scalar $\beta_n$ on
\eqref{eq:hook-cauchy-component}.  On
$\Symm^{n-1}(V)$, the canonical commutation relations give
$\sum_{\alpha=1}^{16}a_\alpha a_\alpha^*=(n+15)I$.
Taking traces, we obtain
\begin{equation}\label{eq:direct-deletion-trace}
 \Tr M_n
 =(n+15)\dim\mathcal D_{n-1}.
\end{equation}
Using $\dim\mathcal D_n=\binom{n+3}{3}^2$ and the dimension formula
\cite[Theorem~6.3(1)]{FultonHarris},
\begin{equation}\label{eq:four-dimensional-hook-dimension}
 \dim\Schur{(n-1,1)}\C^4
 =\frac{(n-1)(n+1)(n+2)}{2},
\end{equation}
we can rewrite \eqref{eq:direct-deletion-trace} as
\begin{equation*}
 (n+15)\binom{n+2}{3}^2
 =n\binom{n+3}{3}^2
 +\beta_n
 \left(\frac{(n-1)(n+1)(n+2)}{2}\right)^2.
\end{equation*}
Solving gives $\beta_n=n/(n-1)$.  Interchanging replicas $3$ and $4$ gives the
crossed statement.
\end{proof}

We now combine the two pairing estimates.  Define the ground-space deletion
operator by
\begin{equation}\label{eq:ground-space-deletion-operator}
 T_n
 =n\,\iota_n^*
  \left(I_V\otimes\mathsf P_{n-1}\right)
  \iota_n.
\end{equation}

\begin{proposition}\label{prop:ground-space-deletion-bound}
For $n\geqslant6$,
\[
 T_n
 \leqslant
 \frac{2n(n-2)}{(n-1)(n-3)}I
 \quad\text{on }
 \mathcal S_n\ominus\mathcal G_n.
\]
\end{proposition}

\begin{proof}
Let $\xi\in\mathcal S_n\ominus\mathcal G_n$.  By
Proposition~\ref{prop:complete-graph-ground-space}, $\xi$ is orthogonal to
both pairing spaces.  Proposition~\ref{prop:single-pairing-deletion-spectrum}
gives
\begin{align*}
 \left\|
  (I_V\otimes\mathsf D_{n-1})
  \iota_n\xi
 \right\|^2
 &=\frac1n\ip{\xi}{M_n\xi}
 \leqslant\frac{1}{n-1}\|\xi\|^2,\\
 \left\|
  (I_V\otimes\mathsf X_{n-1})
  \iota_n\xi
 \right\|^2
 &\leqslant\frac{1}{n-1}\|\xi\|^2.
\end{align*}
Since $n-1\geqslant5$, Proposition~\ref{prop:pairing-principal-angle} gives
\[
 \mathsf P_{n-1}
 \leqslant\frac{n-2}{n-3}
 \left(\mathsf D_{n-1}
      +\mathsf X_{n-1}\right).
\]
Consequently,
\[
 \left\|
  (I_V\otimes\mathsf P_{n-1})
  \iota_n\xi
 \right\|^2
 \leqslant
 \frac{2(n-2)}{(n-1)(n-3)}\|\xi\|^2.
\]
The result follows from \eqref{eq:ground-space-deletion-operator}.
\end{proof}

\subsection{The gap recurrence}

The isometry $\iota_n$ is the natural inclusion obtained by separating
one tensor factor from a symmetric tensor.  Since every edge survives
exactly $n-2$ choices of a deleted vertex, it gives the following identity.

\begin{lemma}\label{lem:two-body-deletion-identity}
Let $n\geqslant3$.  Then
\begin{equation}\label{eq:two-body-deletion-identity}
 n\,\iota_n^*(I_V\otimes R_{n-1})\iota_n=(n-2)R_n.
\end{equation}
\end{lemma}

\begin{proof}
For $\xi,\zeta\in\mathcal S_n$, the matrix element of the
reaction sum after deleting a vertex is independent of the deleted
vertex.  Therefore
\[
 n\langle\iota_n\xi,(I_V\otimes R_{n-1})\iota_n\zeta\rangle
 =\sum_{r=1}^n\sum_{\substack{i<j\\i,j\neq r}}
   \langle\xi,q_{\{i,j\}}\zeta\rangle
 =(n-2)\langle\xi,R_n\zeta\rangle.
\]
This proves the operator identity.
\end{proof}

For $n\geqslant6$, put
$r_n=n(n^2-6n+7)/((n-1)(n-2)(n-3))$.

\begin{proposition}\label{prop:complete-graph-gap-recurrence}
For $n\geqslant6$,
\begin{equation}\label{eq:complete-graph-gap-recurrence}
 \delta_n\geqslant r_n\delta_{n-1}.
\end{equation}
Moreover, $0<r_n<1$.
\end{proposition}

\begin{proof}
On the $(n-1)$-particle sector,
$R_{n-1}\geqslant\delta_{n-1}(\mathsf Q_{n-1}-\mathsf P_{n-1})$.
Applying Lemma~\ref{lem:two-body-deletion-identity} and using the isometry
$\iota_n$, we obtain
\[
 (n-2)R_n
 \geqslant\delta_{n-1}\left(nI-T_n\right).
\]
On $\mathcal S_n\ominus\mathcal G_n$, Proposition~\ref{prop:ground-space-deletion-bound} gives
\[
 \delta_n
 \geqslant\frac{\delta_{n-1}}{n-2}
 \left(
  n-\frac{2n(n-2)}{(n-1)(n-3)}
 \right).
\]
This gives \eqref{eq:complete-graph-gap-recurrence}.

For $n\geqslant6$, the factorisation
$n^2-6n+7=(n-3-\sqrt2)(n-3+\sqrt2)$ shows that
$r_n>0$.  Also, $1-r_n=2(2n-3)/((n-1)(n-2)(n-3))>0$,
so $r_n<1$.
\end{proof}

To obtain the initial bound at $n=5$, we use the following spectrum,
computed in Appendix~\ref{app:complete-three-particle} using the trivial and
standard representations of $\Sn_3$.

\begin{proposition}\label{prop:complete-graph-three-particle-spectrum}
The complete-graph reaction Hamiltonian on three particles has spectrum
\begin{equation}\label{eq:complete-graph-three-particle-spectrum}
 \Spec\bigl(R_3\bigr)
 =\left\{
  0^{[544]},
  \left(\frac34\right)^{[128]},
  \left(\frac32\right)^{[128]},
  \left(\frac94\right)^{[16]}
 \right\}.
\end{equation}
In particular, $\delta_3=3/4$.
\end{proposition}

The following finite-size inequality gives the bounds at $n=4,5$.
For $n\geqslant4$ and $r\in\{1,\ldots,n\}$, define the reaction Hamiltonian
with vertex $r$ deleted by
\[
 R_{K_n\setminus r}
 =\sum_{\substack{1\leqslant i<j\leqslant n\\i,j\neq r}}
  q_{\{i,j\}}
 \quad\text{on }\mathcal H_n\cong V^{\otimes n}.
\]

\begin{lemma}\label{lem:complete-graph-finite-size-inequality}
For $n\geqslant4$, the following inequality holds on $\mathcal H_n$:
\begin{equation}\label{eq:complete-graph-finite-size-operator-inequality}
 \bigl(R_{K_n}\bigr)^2
 -\frac{1}{n-3}\sum_{r=1}^n
  \bigl(R_{K_n\setminus r}\bigr)^2
 +\frac{1}{n-3}R_{K_n}
 \geqslant0.
\end{equation}
Consequently,
\begin{equation}\label{eq:complete-graph-finite-size-gap}
 \delta_n
 \geqslant\frac{(n-2)\delta_{n-1}-1}{n-3}
\end{equation}
whenever the right-hand side is positive.
\end{lemma}

\begin{proof}
Expand the squares in
\eqref{eq:complete-graph-finite-size-operator-inequality}.  Each diagonal
term $(q_e)^2=q_e$ occurs in
$n-2$ deleted Hamiltonians, with total coefficient
$1-(n-2)/(n-3)+1/(n-3)=0$.  Each ordered cross term for intersecting edges
occurs $n-3$ times and also cancels.  A cross term for disjoint edges occurs
$n-4$ times, leaving coefficient $1/(n-3)$.  The corresponding projections
commute, so the remaining terms are positive.  This proves the operator
inequality.

For the gap bound, single out the $r$th particle.  On
$V\otimes\mathcal S_{n-1}$, the deleted Hamiltonian is
$I_V\otimes R_{n-1}$.  Since
$\mathcal S_n\subseteq
V\otimes\mathcal S_{n-1}$, we have
\[
 \bigl(R_{K_n\setminus r}\bigr)^2
 \geqslant\delta_{n-1}R_{K_n\setminus r}
\]
in the sense of quadratic forms on $\mathcal S_n$.
Every edge belongs to $n-2$ deleted complete graphs, so on $\mathcal S_n$,
\[
 \sum_{r=1}^nR_{K_n\setminus r}
 =(n-2)R_n.
\]
Substituting into
\eqref{eq:complete-graph-finite-size-operator-inequality} gives
\[
 \bigl(R_n\bigr)^2
 \geqslant
 \frac{(n-2)\delta_{n-1}-1}{n-3}
 R_n,
\]
on $\mathcal S_n$.  The gap bound follows.
\end{proof}

Starting from $\delta_3=3/4$, we apply
Lemma~\ref{lem:complete-graph-finite-size-inequality} with $n=4$ and $n=5$
to obtain the following bounds.

\begin{corollary}\label{cor:complete-graph-four-five-gaps}
The finite-size gaps satisfy
\[
 \delta_4\geqslant2\delta_3-1=\frac12,
 \qquad
 \delta_5\geqslant\frac{3\delta_4-1}{2}\geqslant\frac14.
\]
\end{corollary}

\begin{proof}[Proof of Theorem~\ref{thm:uniform-complete-graph-gap}]
Since $\delta_2=1$, Proposition~\ref{prop:complete-graph-three-particle-spectrum}
and Corollary~\ref{cor:complete-graph-four-five-gaps} give
$\delta_n\geqslant1/4>c_0$ for $2\leqslant n\leqslant5$.
For $n\geqslant6$, iterating
Proposition~\ref{prop:complete-graph-gap-recurrence} gives
\begin{equation}\label{eq:finite-deletion-product}
 \delta_n
 \geqslant\frac14\prod_{j=6}^nr_j.
\end{equation}
The factors lie in $(0,1)$, so it remains to evaluate their infinite
product.  Using $j^2-6j+7=(j-3-\sqrt2)(j-3+\sqrt2)$, we obtain
\begin{equation}\label{eq:infinite-deletion-product-gamma}
 \prod_{j=6}^{\infty}r_j
 =\frac{12}{5\Gamma(3-\sqrt2)\Gamma(3+\sqrt2)}.
\end{equation}
Indeed, the product up to $N$ is
\[
 \frac{
  N\Gamma(N-2-\sqrt2)\Gamma(3)
  \Gamma(N-2+\sqrt2)\Gamma(4)
 }{
  5\Gamma(3-\sqrt2)\Gamma(N-2)
  \Gamma(3+\sqrt2)\Gamma(N-1)
 },
\]
and taking $N\to\infty$ gives
\eqref{eq:infinite-deletion-product-gamma}.

The Gamma recurrence and Euler's reflection formula give
\[
 \Gamma(3-\sqrt2)\Gamma(3+\sqrt2)
 =-2\frac{\pi\sqrt2}{\sin(\pi\sqrt2)}.
\]
Consequently,
\begin{equation}\label{eq:infinite-deletion-product-value}
 \prod_{j=6}^{\infty}r_j
 =-\frac{6\sin(\pi\sqrt2)}{5\pi\sqrt2}.
\end{equation}
Since $1<\sqrt2<3/2$, this product is positive.  Combining
\eqref{eq:finite-deletion-product} and
\eqref{eq:infinite-deletion-product-value} gives $\delta_n\geqslant c_0$
for $n\geqslant6$, as required.
\end{proof}

\begin{proof}[Proof of Theorem~\ref{thm:all-to-all}]
We have $H_{K_n}=E_{K_n}+R_{K_n}$.
Both sums commute with permutations of the sites, so they
preserve $\mathcal S_n$ and its orthogonal complement.
On $\mathcal S_n$, the interchange term vanishes and
Theorem~\ref{thm:uniform-complete-graph-gap} gives
\[
 H_{K_n}
 \geqslant c_0
 \left(\mathsf Q_n-\mathsf P_n\right).
\]
On the orthogonal complement, Aldous' theorem on $K_n$ gives
\[
 E_{K_n}
 \geqslant\frac n2\left(I-\mathsf Q_n\right).
\]
Since the reaction term is positive and $n/2\geqslant c_0$, this proves
\eqref{eq:complete-graph-moment-gap}.

\end{proof}

\section{A Poisson equation on unordered pairs}\label{sec:edge-defect}

We solve a Poisson equation for the centred indicator of the path edges on
the permutation module of unordered pairs.  The resulting identities for
the Moore--Penrose inverse will be used to bound the off-diagonal reaction
block.  Throughout this section, let $n\geqslant4$ and write
$[n]=\{1,\ldots,n\}$.

Let $\mathcal W_n=\ell^2(\binom{[n]}{2})$, with orthonormal basis
$\mathbf e_{\{i,j\}}$, $1\leqslant i<j\leqslant n$.  The symmetric group
acts by $\pi_n(\sigma)\mathbf e_{\{i,j\}}
=\mathbf e_{\{\sigma(i),\sigma(j)\}}$.  Young's rule \cite[Theorem~2.11.2]{Sagan} gives the multiplicity-free
decomposition
\begin{equation}\label{eq:edge-module-decomposition}
 \mathcal W_n
 \cong
 \Specht{(n)}\oplus\Specht{(n-1,1)}\oplus\Specht{(n-2,2)}.
\end{equation}
We denote these summands by $\mathcal W_0$, $\mathcal W_1$ and
$\mathcal W_2$, respectively, and write $\Pi_j$ for the orthogonal
projection onto $\mathcal W_j$.

Let $\mathbf d$ be the centred indicator of the path edges, given by
\begin{equation}\label{eq:path-edge-defect}
 \mathbf d
 =\sum_{i=1}^{n-1}\mathbf e_{\{i,i+1\}}
  -\frac{2}{n}\sum_{1\leqslant i<j\leqslant n}\mathbf e_{\{i,j\}}.
\end{equation}
Since $2/n=(n-1)/\binom n2$, this vector is orthogonal to
$\mathcal W_0$.  For $j=1,2$, let $\mathbf d_j=\Pi_j\mathbf d$.

\begin{proposition}\label{prop:path-edge-component-norms}
We have
\begin{equation}\label{eq:path-edge-component-norms}
 \left\|\mathbf d_1\right\|^2=\frac{2}{n},
 \qquad
 \left\|\mathbf d_2\right\|^2=n-3.
\end{equation}
In particular,
\begin{equation}\label{eq:path-edge-defect-norm}
 \left\|\mathbf d\right\|^2
 =\frac{(n-1)(n-2)}{n}.
\end{equation}
\end{proposition}

\begin{proof}
Let $\mathcal V_n$ be the standard module from
Section~\ref{sec:transport}.  The equivariant map
\begin{equation}\label{eq:vertex-to-edge-map}
 J_n:\mathcal V_n\longrightarrow\mathcal W_n,
 \qquad
 (J_n\mathbf u)_{\{i,j\}}=u_i+u_j
\end{equation}
has range $\mathcal W_1$ and satisfies
\begin{equation}\label{eq:vertex-to-edge-norm}
 \|J_n\mathbf u\|^2
 =\sum_{i<j}|u_i+u_j|^2
 =(n-2)\sum_{i=1}^n|u_i|^2.
\end{equation}
The row sums satisfy $\sum_{j\neq i}(J_n\mathbf u)_{\{i,j\}}=(n-2)u_i$,
whereas
\[
 \sum_{j\neq i}d_{\{i,j\}}
 =\begin{cases}
  -(n-2)/n,&i=1,n,\\[1mm]
  2/n,&2\leqslant i\leqslant n-1.
 \end{cases}
\]
The adjoint $J_n^*$ sends an edge function to the zero-sum part of its
row-sum vector.  Thus $\mathbf d_1=J_n\mathbf u$, where
\begin{equation}\label{eq:standard-vertex-defect}
 (\mathbf u)_i
 =\begin{cases}
  -1/n,&i=1,n,\\[1mm]
  2/[n(n-2)],&2\leqslant i\leqslant n-1.
 \end{cases}
\end{equation}
Since $\|\mathbf u\|^2=2/[n(n-2)]$,
\eqref{eq:vertex-to-edge-norm} gives
$\|\mathbf d_1\|^2=2/n$.
Also,
\[
 \|\mathbf d\|^2
 =(n-1)\left(\frac{n-2}{n}\right)^2
 +\left(\binom n2-(n-1)\right)\frac{4}{n^2}
 =\frac{(n-1)(n-2)}{n}.
\]
Subtracting $\|\mathbf d_1\|^2$ gives
$\|\mathbf d_2\|^2=n-3$.
\end{proof}

To solve the Poisson equation, define
\begin{equation}\label{eq:edge-exclusion-laplacian}
 L_n^{(2)}
 =\frac12\sum_{k=1}^{n-1}
 \left(I-\pi_n(\transp{k}{k+1})\right).
\end{equation}
This is the positive Laplacian of the two-particle simple exclusion process
on the path, with jump rate $1/2$ across each available edge.  Its kernel is
$\mathcal W_0$.  For $1\leqslant i<j\leqslant n$, set
\begin{equation}\label{eq:edge-poisson-coefficients}
 g_{ij}
 =\frac{
   \left(i-\frac{n+1}{2}\right)^2
   +\left(j-\frac{n+1}{2}\right)^2
  }{n}
  -(j-i)
  +\frac{(n+1)^2}{6n},
\end{equation}
and let
$\mathbf g
=\sum_{i<j}g_{ij}\mathbf e_{\{i,j\}}$.

\begin{proposition}\label{prop:edge-poisson-equation}
The vector $\mathbf g$ has zero mean and satisfies
\begin{equation}\label{eq:edge-poisson-equation}
 L_n^{(2)}\mathbf g
 =\mathbf d.
\end{equation}
Moreover,
\begin{equation}\label{eq:total-edge-resolvent}
 \frac{
  \ip{\mathbf d}
      {(L_n^{(2)})^+\mathbf d}
 }{
  \|\mathbf d\|^2
 }
 =\frac{n-1}{3}.
\end{equation}
\end{proposition}

\begin{proof}
We may write
\[
 g_{ij}
 =\frac{i(i-1)}{n}
  +\frac{(j-1)(j-2n)}{n}
  +\frac{2(n^2-n+1)}{3n}.
\]
If $j>i+1$, each coordinate contributes $-1/n$ to
$(L_n^{(2)}\mathbf g)_{\{i,j\}}$;
the same value holds at the endpoints by the one-sided difference.  Hence
$(L_n^{(2)}\mathbf g)_{\{i,j\}}=-2/n$.
For adjacent pairs with $2\leqslant i\leqslant n-2$, the two moves towards the
occupied site are blocked, so
\[
 \bigl(L_n^{(2)}\mathbf g\bigr)_{\{i,i+1\}}
 =\frac12(g_{i,i+1}-g_{i-1,i+1})
  +\frac12(g_{i,i+1}-g_{i,i+2})
 =\frac{n-2}{n}.
\]
At $i=1$ and $i=n-1$, the single available move gives the same value.
This proves \eqref{eq:edge-poisson-equation}.

Using
\[
 \sum_{k=1}^n\left(k-\frac{n+1}{2}\right)^2
 =\frac{n(n^2-1)}{12},
 \qquad
 \sum_{i<j}(j-i)=\frac{n(n^2-1)}{6}
\]
in \eqref{eq:edge-poisson-coefficients}, we obtain $\sum_{i<j}g_{ij}=0$.
Thus $\mathbf g
=(L_n^{(2)})^+\mathbf d$ and
\[
 \ip{\mathbf d}{\mathbf g}
 =\sum_{i=1}^{n-1}g_{i,i+1}
 =\frac{(n-1)^2(n-2)}{3n}.
\]
Now \eqref{eq:path-edge-defect-norm} gives
\eqref{eq:total-edge-resolvent}.
\end{proof}

For the coupling estimate, we need the same identity on each non-trivial
irreducible component.

\begin{proposition}\label{prop:componentwise-edge-resolvent}
For $j=1,2$, we have
\begin{equation}\label{eq:componentwise-edge-resolvent}
 \ip{\mathbf d_j}
     {(L_n^{(2)})^+\mathbf d_j}
 =\frac{n-1}{3}\|\mathbf d_j\|^2.
\end{equation}
\end{proposition}

\begin{proof}
Each summand in \eqref{eq:edge-module-decomposition} is invariant under
$L_n^{(2)}$.  By Proposition~\ref{prop:edge-poisson-equation}, it
therefore suffices to prove the identity on the standard component.

On $\mathcal V_n$, let
\[
 L_n^{(1)}
 =\frac12\sum_{k=1}^{n-1}
 \left(I-\rho(\transp{k}{k+1})\right).
\]
The map $J_n$ satisfies
\begin{equation}\label{eq:vertex-edge-laplacian-intertwining}
 L_n^{(2)}J_n=J_nL_n^{(1)}.
\end{equation}
Let $\mathbf h=(L_n^{(1)})^+\mathbf u$, where
$\mathbf u$ is given by \eqref{eq:standard-vertex-defect}.
The endpoint and interior equations give
\[
 (\mathbf h)_{k+1}-(\mathbf h)_k
 =\frac{2}{n}-\frac{4(k-1)}{n(n-2)},
 \qquad 1\leqslant k\leqslant n-1.
\]
Consequently,
\[
 \ip{\mathbf u}{(L_n^{(1)})^+\mathbf u}
 =\frac12\sum_{k=1}^{n-1}|(\mathbf h)_{k+1}-(\mathbf h)_k|^2
 =\frac{2(n-1)}{3n(n-2)}
 =\frac{n-1}{3}\|\mathbf u\|^2.
\]
Equations \eqref{eq:vertex-to-edge-norm} and
\eqref{eq:vertex-edge-laplacian-intertwining} give the identity for
$\mathbf d_1$.  Orthogonality and
\eqref{eq:total-edge-resolvent} give the remaining component.
\end{proof}

Thus the Rayleigh quotient of $(L_n^{(2)})^+$ at each $\mathbf d_j$ is
$(n-1)/3$, whereas its operator norm has order $n^2$.

\section{The off-diagonal reaction block}\label{sec:slow-fast}

We use the identities from Section~\ref{sec:edge-defect} to estimate the
off-diagonal reaction block.  Throughout this section, let $n\geqslant4$.  Write
\[
 A_n=\mathsf Q_nR_{P_n}\mathsf Q_n,
 \qquad
 B_n=(I-\mathsf Q_n)R_{P_n}\mathsf Q_n.
\]
Thus $A_n$ is the compression of the path reaction Hamiltonian to
$\mathcal S_n$, and $B_n$ is its off-diagonal block from
$\mathcal S_n$ to $\mathcal S_n^\perp$.

\begin{proposition}\label{prop:symmetric-reaction-compression}
On $\mathcal S_n$, we have
\begin{equation}\label{eq:symmetric-reaction-compression}
 A_n=\frac{2}{n}R_n.
\end{equation}
In particular,
\begin{equation}\label{eq:symmetric-reaction-gap}
 A_n\geqslant\frac{2c_0}{n}
 \left(\mathsf Q_n-\mathsf P_n\right).
\end{equation}
\end{proposition}

\begin{proof}
For $\xi,\zeta\in\mathcal S_n$, the inner product
$\ip{\xi}{q_{\{i,j\}}\zeta}$ is independent of
$\{i,j\}$.  Since the path and complete graph have $n-1$ and $\binom n2$
edges, respectively,
\[
 \ip{\xi}{R_{P_n}\zeta}
 =\frac{2}{n}\ip{\xi}{R_n\zeta}.
\]
This proves \eqref{eq:symmetric-reaction-compression}, and
Theorem~\ref{thm:uniform-complete-graph-gap} gives
\eqref{eq:symmetric-reaction-gap}.
\end{proof}

Fix $\xi\in\mathcal S_n$ and define
\[
 \Phi_{\xi}:\mathcal W_n\longrightarrow\mathcal H_n,
 \qquad
 \Phi_{\xi}\mathbf e_{\{i,j\}}
 =q_{\{i,j\}}\xi.
\]
Since $\xi$ is fixed by simultaneous permutations in all four replicas,
\begin{equation}\label{eq:reaction-edge-to-state-equivariance}
 U_\sigma\Phi_{\xi}
 =\Phi_{\xi}\pi_n(\sigma)
 \qquad(\sigma\in\Sn_n).
\end{equation}
Averaging the path-edge sum over $\Sn_n$ gives its complete-graph average.
Thus \eqref{eq:path-edge-defect} gives
\begin{equation}\label{eq:slow-fast-block-as-edge-defect}
 B_n\xi=\Phi_{\xi}\mathbf d.
\end{equation}

\begin{proposition}\label{prop:slow-fast-contraction}
We have
\begin{equation}\label{eq:slow-fast-contraction}
 B_n^*B_n\leqslant A_n.
\end{equation}
\end{proposition}

\begin{proof}
The operator $(\Phi_{\xi})^*
\Phi_{\xi}$ commutes with $\pi_n(\Sn_n)$.  By Schur's lemma
and \eqref{eq:edge-module-decomposition}, it is scalar on each irreducible
summand.  Hence, for $j=1,2$,
\begin{equation}\label{eq:edge-component-image-norm}
 \|\Phi_{\xi}\mathbf d_j\|^2
 =\frac{\|\mathbf d_j\|^2}{\dim\mathcal W_j}
 \Tr\!\left[
  \Pi_j
  (\Phi_{\xi})^*\Phi_{\xi}
 \right].
\end{equation}
By Proposition~\ref{prop:path-edge-component-norms},
\[
 \frac{\|\mathbf d_1\|^2}
      {\dim\mathcal W_1}
 =\frac{2}{n(n-1)}\leqslant\frac2n,
 \qquad
 \frac{\|\mathbf d_2\|^2}
      {\dim\mathcal W_2}=\frac2n.
\]
The two image components lie in inequivalent isotypic subspaces and are
therefore orthogonal.  Summing \eqref{eq:edge-component-image-norm}, we obtain
\begin{align*}
 \|\Phi_{\xi}\mathbf d\|^2
 &\leqslant\frac2n\Tr\!\left[
   (\Phi_{\xi})^*\Phi_{\xi}\right]
 =\frac2n\sum_{i<j}\|q_{\{i,j\}}\xi\|^2\\
 &=(n-1)\|q_{\{1,2\}}\xi\|^2
  =\ip{\xi}{A_n\xi}.
\end{align*}
Now \eqref{eq:slow-fast-block-as-edge-defect} proves the result.
\end{proof}

The equivariant map also intertwines the interchange Hamiltonian and the edge
Laplacian.

\begin{lemma}\label{lem:interchange-edge-intertwining}
For every $\xi\in\mathcal S_n$, we have
\[
 E_{P_n}\Phi_{\xi}
 =\Phi_{\xi}L_n^{(2)}.
\]
\end{lemma}

\begin{proof}
For $\sigma_k=\transp{k}{k+1}$,
\eqref{eq:reaction-edge-to-state-equivariance} gives
$U_{\sigma_k}q_e\xi
=q_{\sigma_k e}\xi$.  Therefore
\[
 E_{P_n}\Phi_{\xi}\mathbf e_e
 =\frac12\sum_{k=1}^{n-1}
  \bigl(q_e\xi
       -q_{\sigma_k e}\xi\bigr)
 =\Phi_{\xi}L_n^{(2)}\mathbf e_e.
\]
The result follows by linearity.
\end{proof}

\begin{proposition}\label{prop:exact-slow-fast-resolvent}
On $\mathcal S_n$, we have
\begin{equation}\label{eq:exact-slow-fast-resolvent}
 B_n^*(E_{P_n})^+B_n
 =\frac{n-1}{3}B_n^*B_n.
\end{equation}
Consequently,
\begin{equation}\label{eq:slow-fast-resolvent-bound}
 B_n^*(E_{P_n})^+B_n
 \leqslant\frac{n-1}{3}A_n.
\end{equation}
\end{proposition}

\begin{proof}
Fix $\xi\in\mathcal S_n$.  The image of each non-trivial
edge component is orthogonal to
$\mathcal S_n=\Ker E_{P_n}$.
For $j=1,2$, Lemma~\ref{lem:interchange-edge-intertwining} gives
\[
 (E_{P_n})^+
 \Phi_{\xi}\mathbf d_j
 =\Phi_{\xi}
  (L_n^{(2)})^+\mathbf d_j.
\]
Since $(\Phi_{\xi})^*\Phi_{\xi}$ is scalar
on $\mathcal W_j$, Proposition~\ref{prop:componentwise-edge-resolvent}
implies
\[
 \ip{\Phi_{\xi}\mathbf d_j}{(E_{P_n})^+\Phi_{\xi}\mathbf d_j}
 =\frac{n-1}{3}\|\Phi_{\xi}\mathbf d_j\|^2.
\]
Summing over the two orthogonal components and using
\eqref{eq:slow-fast-block-as-edge-defect} proves
\eqref{eq:exact-slow-fast-resolvent} as a quadratic-form identity, and hence
as an operator identity.  Proposition~\ref{prop:slow-fast-contraction} gives
\eqref{eq:slow-fast-resolvent-bound}.
\end{proof}

We use the following dual form in the two-scale estimate.

\begin{proposition}\label{prop:dual-slow-fast-schur}
We have
\begin{equation}\label{eq:dual-slow-fast-schur}
 B_nA_n^+B_n^*\leqslant\frac{n-1}{3}E_{P_n}.
\end{equation}
\end{proposition}

\begin{proof}
By Proposition~\ref{prop:slow-fast-contraction} and the definition of $B_n$,
\[
 \Ker A_n\subseteq\Ker B_n,
 \qquad
 \Ran B_n\subseteq\mathcal S_n^\perp
 =\Ran E_{P_n}.
\]
Equation \eqref{eq:slow-fast-resolvent-bound} gives
\[
 \|(E_{P_n})^{+1/2}B_nA_n^{+1/2}\|^2
 \leqslant\frac{n-1}{3}.
\]
These kernel and range inclusions imply
\[
 B_n=(E_{P_n})^{1/2}
 \bigl[(E_{P_n})^{+1/2}B_nA_n^{+1/2}\bigr]A_n^{1/2}.
\]
Substituting this expression into $B_nA_n^+B_n^*$ gives
\eqref{eq:dual-slow-fast-schur}.
\end{proof}

\section{A two-scale Poincar\'e inequality}\label{sec:two-scale}

We combine the interchange gap, the symmetric reaction gap and the coupling
estimate in a two-scale Poincar\'e inequality.  We then apply it to prove
Theorem~\ref{thm:main}.

Let $H=E+R$ be a sum of positive semidefinite
operators on a finite-dimensional Hilbert space.  Let
$\mathsf Q$ and $\mathsf P$ be the orthogonal
projections onto $\Ker E$ and $\Ker H$, respectively, so that
$\mathsf P\leqslant\mathsf Q$.  Write
\[
 A=\mathsf Q R\mathsf Q,
 \qquad B=(I-\mathsf Q)R\mathsf Q,
 \qquad C=(I-\mathsf Q)R(I-\mathsf Q).
\]

\begin{proposition}\label{prop:two-scale-poincare}
Let $\alpha,\beta,\eta>0$ satisfy
\begin{align}
 E
 &\geqslant\alpha(I-\mathsf Q),
 \label{eq:abstract-fast-coercivity}\\
 A&\geqslant\beta
   (\mathsf Q-\mathsf P),
 \label{eq:abstract-slow-coercivity}\\
 BA^+B^*&\leqslant\eta E.
 \label{eq:abstract-coupling-control}
\end{align}
Then
\begin{equation}\label{eq:two-scale-poincare}
 H\geqslant
 \left(\frac{1+\eta}{\beta}
       +\frac{1}{\alpha}\right)^{-1}
 (I-\mathsf P).
\end{equation}
\end{proposition}

\begin{proof}
For an arbitrary vector $\psi$, write
$x=\mathsf Q\psi$ and
$y=(I-\mathsf Q)\psi$, and set
\[
 e
 =\ip{\psi}{E\psi},
 \qquad
 r
 =\ip{\psi}{R\psi}.
\]
By \eqref{eq:abstract-fast-coercivity},
\begin{equation}\label{eq:abstract-fast-component-bound}
 \|y\|^2
 \leqslant\frac{e}{\alpha}.
\end{equation}

Positivity of the reaction block matrix gives
$\Ker A\subseteq\Ker B$ and $C\geqslant BA^+B^*$.  Set
\[
 a=\ip{x}{Ax},
 \qquad
 b=\ip{y}{BA^+B^*y}.
\]
Cauchy--Schwarz gives
$|\ip{y}{Bx}|
\leqslant\sqrt{ab}$, and hence
\[
 r
 \geqslant a-2\sqrt{ab}+b=(\sqrt a-\sqrt b)^2.
\]
Using \eqref{eq:abstract-coupling-control}, we obtain
\[
 \sqrt a\leqslant
 \sqrt{r}
 +\sqrt{\eta e},
 \qquad
 a\leqslant(1+\eta)\ip{\psi}{H\psi}.
\]
Now \eqref{eq:abstract-slow-coercivity} gives
\begin{equation}\label{eq:abstract-slow-component-bound}
 \|(\mathsf Q-\mathsf P)\psi\|^2
 \leqslant\frac{1+\eta}{\beta}\ip{\psi}{H\psi}.
\end{equation}
The vectors on the left-hand sides of
\eqref{eq:abstract-fast-component-bound} and
\eqref{eq:abstract-slow-component-bound} are orthogonal.  Adding these bounds,
with $e\leqslant\ip{\psi}{H\psi}$,
proves \eqref{eq:two-scale-poincare}.
\end{proof}

\begin{proof}[Proof of Theorem~\ref{thm:main}]
For $n=2$, the Hamiltonian is the non-zero projection
$h_{\{1,2\}}$, so $\Delta(H_2)=1$.
Appendix~\ref{app:path-three-particle} gives $\Delta(H_3)=3/8$.
Suppose now that $n\geqslant4$.  Apply
Proposition~\ref{prop:two-scale-poincare} with
\[
 E=E_{P_n},
 \quad R=R_{P_n},
 \quad\mathsf Q=\mathsf Q_n,
 \quad\mathsf P=\mathsf P_n.
\]
Propositions \ref{prop:interchange-gap},
\ref{prop:symmetric-reaction-compression} and
\ref{prop:dual-slow-fast-schur} give, respectively,
\[
 \alpha=1-\cos\frac{\pi}{n},
 \qquad\beta=\frac{2c_0}{n},
 \qquad\eta=\frac{n-1}{3}.
\]
Thus
\[
 \Delta(H_n)\geqslant
 \left(\frac{n(n+2)}{6c_0}
       +\frac{1}{1-\cos(\pi/n)}\right)^{-1}.
\]
Proposition~\ref{prop:diffusive-upper-bound} gives the upper bound in
\eqref{eq:main-refined-bounds} and the spectral inclusion
\eqref{eq:path-laplacian-branch}.

Since $n+2\leqslant3n/2$ and $1-\cos(\pi/n)\geqslant2/n^2$ for
$n\geqslant4$, we have
\[
 \Delta(H_n)\geqslant
 \left(\frac{n^2}{4c_0}+\frac{n^2}{2}\right)^{-1}
 =\frac{c_*}{n^2}.
\]
The values for $n=2,3$ satisfy the same lower bound.  Finally,
$1-\cos t\leqslant t^2/2$ gives the stated upper bound.
\end{proof}

\section{Circuit convergence}\label{sec:circuits}

The spectral bounds give convergence estimates for
sequential, brickwork, periodic and all-to-all circuits.  All logarithms
are natural.  The mixing time is the least non-negative number of steps
for which the stated error is at most $\varepsilon$.

\subsection{Sequential circuits and CP-order approximation}

For an ensemble $\mathcal E$, write $\mathsf T_2^{\mathcal E}$
for its second-moment channel and $\mathsf T_2^{\Univ_\times}$ for the Haar
second-moment channel on $\Univ_\times$.  We write $\mathsf A\leqslant_{\rm cp}
\mathsf B$ when $\mathsf B-\mathsf A$ is completely positive, and denote by
$c_{\mathrm{cp}}(\mathcal E,2)$ the least $\varepsilon\geqslant0$ such that
\begin{equation}\label{eq:cp-order-definition}
 (1-\varepsilon)\mathsf T_2^{\Univ_\times}
 \leqslant_{\rm cp}\mathsf T_2^{\mathcal E}
 \leqslant_{\rm cp}(1+\varepsilon)\mathsf T_2^{\Univ_\times}.
\end{equation}
This is the notion of approximate designs of
\cite[Definition~2 and Eq.~(4)]{LiEtAl}.  The convolution
$\mathcal E^{*\ell}$ is the distribution of a product of $\ell$ independent
unitaries drawn from $\mathcal E$.

\begin{corollary}\label{cor:moment-mixing}
For $n\geqslant2$ and every integer $\ell\geqslant1$,
\begin{equation}\label{eq:exact-moment-mixing}
 \left\|\mathsf M_n^\ell-\mathsf P_n\right\|
 =\left(1-\frac{\Delta(H_n)}{n-1}\right)^\ell.
\end{equation}
Thus $\mathsf M_n$ has spectral gap $\Theta(n^{-3})$.  For $n\geqslant3$, its
operator-norm mixing time to accuracy $0<\varepsilon<1$ is
$\Theta(1+n^3\log\varepsilon^{-1})$; for $n=2$, one step gives exact mixing.
For $n\geqslant5$ and $0<\varepsilon<1$, every integer $\ell$ satisfying
\begin{equation}\label{eq:cp-design-depth}
 \ell\geqslant
 \frac{n^2(n-1)}{c_*}
 \left(4n\log2+\log\frac1\varepsilon\right),
\end{equation}
produces an $\varepsilon$-approximate $\mathrm{SU}(2)$-symmetric $2$-design.
In particular, the required number of random gates is
$O(n^4+n^3\log\varepsilon^{-1})$.
\end{corollary}

\begin{proof}
Each $\mathsf T_e^{(2)}$ is an orthogonal projection.  Hence $\mathsf M_n$ is
positive semidefinite, and its eigenspace for the eigenvalue $1$ is $\Ker H_n$.
The spectral theorem and \eqref{eq:intro-moment-operator} give
\eqref{eq:exact-moment-mixing}.  For $n\geqslant5$, the fixed space agrees
with that of all $\mathrm{SU}(2)$-symmetric unitaries
\cite[Theorem~8 and Eqs.~(C38)--(C40)]{LiEtAl}.
By the CP-order comparison \cite[Eq.~(15)]{LiEtAl}, it suffices that
$2^{4n}\|\mathsf M_n^\ell-\mathsf P_n\|\leqslant\varepsilon$.
Theorem~\ref{thm:main} gives this bound whenever
\eqref{eq:cp-design-depth} holds.
\end{proof}

Let $q_n=1-\Delta(H_n)/(n-1)$.

\begin{corollary}\label{cor:sharp-cp-mixing}
For $n\geqslant5$ and every integer
$\ell\geqslant0$, the $\ell$-fold convolution of the one-step ensemble
satisfies
\begin{equation}\label{eq:two-sided-cp-mixing}
 2^{-(2n+1)}q_n^\ell
 \leqslant c_{\mathrm{cp}}(\mathcal E_n^{*\ell},2)
 \leqslant 2^{4n}q_n^\ell.
\end{equation}
Consequently, for every fixed $\eta>0$, the CP-order mixing time is
\begin{equation}\label{eq:sharp-cp-regime}
 \Theta\!\left(n^3\log\varepsilon^{-1}\right)
 \qquad\text{whenever}\qquad
 \log\varepsilon^{-1}\geqslant(2\log2+\eta)n.
\end{equation}
In particular, it is $\Theta(n^4)$ for
$\varepsilon=\exp(-\alpha n)$ with any fixed $\alpha>2\log2$.
\end{corollary}

\begin{proof}
By \eqref{eq:exact-moment-mixing}, the operator-norm error of
$\mathcal E_n^{*\ell}$ is $q_n^\ell$.  Applying the norm comparisons in
\cite[Lemma~B1, Eq.~(B10)]{LiEtAl} with system dimension $2^n$ and moment order two
gives \eqref{eq:two-sided-cp-mixing}.  Theorem~\ref{thm:main} gives
$-\log q_n=\Theta(n^{-3})$.  If
$\log\varepsilon^{-1}\geqslant(2\log2+\eta)n$, both logarithmic terms
obtained from \eqref{eq:two-sided-cp-mixing} have order
$\log\varepsilon^{-1}$, proving the assertion.
\end{proof}

\subsection{Brickwork circuits}

One brickwork period consists of an odd layer followed by an even layer.
Let
\[
 \mathsf\Pi_{\mathrm{odd}}
 =\prod_{\substack{1\leqslant i<n\\ i\ \mathrm{odd}}}
   \mathsf T_{\{i,i+1\}}^{(2)},
 \qquad
 \mathsf\Pi_{\mathrm{even}}
 =\prod_{\substack{1\leqslant i<n\\ i\ \mathrm{even}}}
   \mathsf T_{\{i,i+1\}}^{(2)},
\]
where the eSWAP angle on each bond is averaged independently.  The
ensemble of one odd--even brickwork period is denoted by $\mathcal B_n$, and
its second-moment operator is
$\mathsf B_n=\mathsf\Pi_{\mathrm{even}}\mathsf\Pi_{\mathrm{odd}}$.
Set
$b_n
=\|\mathsf B_n-\mathsf P_n\|$.

\begin{corollary}\label{cor:brickwork}
For the open eSWAP brickwork circuit and every $n\geqslant4$,
\begin{equation}\label{eq:brickwork-gap}
 \begin{aligned}
  \cos\frac{\pi}{n}
  &\leqslant b_n
  \leqslant\left(1+\frac{\Delta(H_n)}4\right)^{-1/2},\\
  \frac{c_*}{4n^2+c_*}
  &\leqslant1-b_n^2
  \leqslant\sin^2\frac{\pi}{n}
  \leqslant\frac{\pi^2}{n^2}.
 \end{aligned}
\end{equation}
For every integer $p\geqslant1$,
\begin{equation}\label{eq:brickwork-exact-power-law}
 \left\|\mathsf B_n^p
       -\mathsf P_n\right\|
 =b_n^{2p-1}.
\end{equation}
Consequently, for $0<\varepsilon<1$, the operator-norm mixing time is
$\Theta(1+n^2\log\varepsilon^{-1})$ periods.  If $n\geqslant5$, then for
every integer $p\geqslant1$,
\begin{equation}\label{eq:brickwork-cp-comparison}
 2^{-(2n+1)}b_n^{2p-1}
 \leqslant c_{\mathrm{cp}}(\mathcal B_n^{*p},2)
 \leqslant2^{4n}b_n^{2p-1}.
\end{equation}
Thus the parallel depth required for an $\varepsilon$-approximate
$\mathrm{SU}(2)$-symmetric $2$-design is
$O(n^3+n^2\log\varepsilon^{-1})$.  For every fixed $\eta>0$, it is
$\Theta(n^2\log\varepsilon^{-1})$ whenever
$\log\varepsilon^{-1}\geqslant(2\log2+\eta)n$, and hence
$\Theta(n^3)$ for $\varepsilon=e^{-\alpha n}$ with fixed
$\alpha>2\log2$.
\end{corollary}

\begin{proof}
The detectability lemma gives
$b_n^2\leqslant(1+\Delta(H_n)/4)^{-1}$
\cite[Eq.~(44)]{LiEtAl}.  For the lower bound on $b_n$, restrict
the two layers to one of the tensor products in
Proposition~\ref{prop:diffusive-upper-bound}.  Each local twirl becomes the
pair-averaging projection
$\tfrac12(I+\rho(\transp{i}{i+1}))$.

We compute the principal angles in the natural permutation module $\C^n$.
For $n=2m$, choose the following orthonormal bases of the odd- and even-layer
ranges:
\begin{align*}
 u_a&=\frac{\mathbf e_{2a-1}+\mathbf e_{2a}}{\sqrt2}
       \quad(1\leqslant a\leqslant m),\\
 v_0&=\mathbf e_1,\qquad
 v_a=\frac{\mathbf e_{2a}+\mathbf e_{2a+1}}{\sqrt2}
       \quad(1\leqslant a<m),\qquad v_m=\mathbf e_{2m}.
\end{align*}
Let $K=(\langle v_b,u_a\rangle)\in\C^{(m+1)\times m}$.  Then $4K^*K$ is
tridiagonal, with off-diagonal entries $1$, endpoint diagonal entries $3$
and all other diagonal entries $2$.
For $n=2m+1$, choose
\begin{align*}
 u_a&=\frac{\mathbf e_{2a-1}+\mathbf e_{2a}}{\sqrt2}
       \ (1\leqslant a\leqslant m),
 &u_{m+1}&=\mathbf e_{2m+1},\\
 v_0&=\mathbf e_1,
 &v_a&=\frac{\mathbf e_{2a}+\mathbf e_{2a+1}}{\sqrt2}
       \ (1\leqslant a\leqslant m).
\end{align*}
Now $K\in\C^{(m+1)\times(m+1)}$, and $4K^*K$ has diagonal entries
$3,2,\ldots,2$ and off-diagonal entries $1$, except for the last, which is
$\sqrt2$.  Expanding the characteristic determinants of these tridiagonal
matrices gives the usual three-term recurrence and hence
\[
 \det(\mu I-K^*K)
 =\prod_{k=0}^{\lfloor(n-1)/2\rfloor}
   \left(\mu-\cos^2\frac{k\pi}{n}\right).
\]
The singular values are therefore $\cos(k\pi/n)$ for
$0\leqslant k\leqslant\lfloor(n-1)/2\rfloor$.
The value at $k=0$ belongs to the constant vector.  On its orthogonal
complement $\mathcal V_n$, the largest singular value is
$\cos(\pi/n)$.  Since this subspace is orthogonal to $\Ker H_n$, we obtain
$b_n\geqslant\cos(\pi/n)$.
Together with Theorem~\ref{thm:main}, this proves
\eqref{eq:brickwork-gap}.

The two layer twirls are orthogonal projections with common range
$\Ker H_n$.  By the two-projection decomposition \cite[Theorem~2]{Halmos}, on a block with principal-angle cosine
$c$, the $p$th power of their ordered product has norm $c^{2p-1}$.
Taking the maximum over the blocks proves
\eqref{eq:brickwork-exact-power-law}.  Equation \eqref{eq:brickwork-gap}
then gives the operator-norm mixing time.

For $n\geqslant5$, the fixed space agrees with the Haar second-moment fixed
space on $\Univ_\times$.  Applying \cite[Lemma~B1, Eq.~(B10)]{LiEtAl} to
\eqref{eq:brickwork-exact-power-law} gives
\eqref{eq:brickwork-cp-comparison}.  The upper bound in
\eqref{eq:brickwork-gap} gives $O(n^3+n^2\log\varepsilon^{-1})$ periods.
If $\log\varepsilon^{-1}\geqslant(2\log2+\eta)n$, the logarithmic factors
in \eqref{eq:brickwork-cp-comparison} have order $\log\varepsilon^{-1}$.
Both bounds in \eqref{eq:brickwork-gap} then give the asserted matching order.
\end{proof}

\subsection{The periodic chain}

The path estimate also gives a diffusive gap for the periodic chain studied in
\cite[Fig.~1]{LiEtAl}.  Let $C_n$ be the cycle, denote the corresponding
one-step ensemble by $\mathcal E_{C_n}$, and set
\[
 \mathsf M_{C_n}=I-\frac1nH_{C_n},
 \qquad
 \gamma_n=\Delta(H_{C_n}),
 \qquad
 t_n=1-\frac{\gamma_n}{n}.
\]

\begin{corollary}\label{cor:periodic-chain}
For every $n\geqslant3$,
\begin{equation}\label{eq:periodic-chain-gap}
 \frac{c_*}{n^2}
 \leqslant\gamma_n
 \leqslant1-\cos\frac{2\pi}{n}.
\end{equation}
The entire non-zero spectrum of one half of the cycle Laplacian is
present:
\begin{equation}\label{eq:cycle-laplacian-branch}
 \left\{1-\cos\frac{2\pi k}{n}:1\leqslant k\leqslant n-1\right\}
 \subseteq\Spec\bigl(H_{C_n}\bigr).
\end{equation}
On each of the four tensor products in
Proposition~\ref{prop:diffusive-upper-bound}, these eigenvalues have the same
multiplicities as those of one half of the cycle Laplacian.
Moreover, for every integer $\ell\geqslant0$,
\begin{equation}\label{eq:periodic-exact-mixing}
 \left\|\mathsf M_{C_n}^\ell
       -\mathsf P_n\right\|
 =t_n^\ell.
\end{equation}
Consequently, $\gamma_n=\Theta(n^{-2})$ and
$\operatorname{gap}(\mathsf M_{C_n})=\Theta(n^{-3})$.
Thus, for $0<\varepsilon<1$, the operator-norm mixing time is
$\Theta(1+n^3\log\varepsilon^{-1})$.  For $n\geqslant5$,
\begin{equation}\label{eq:periodic-cp-comparison}
 2^{-(2n+1)}t_n^\ell
 \leqslant c_{\mathrm{cp}}\bigl((\mathcal E_{C_n})^{*\ell},2\bigr)
 \leqslant2^{4n}t_n^\ell.
\end{equation}
Hence the periodic chain forms an $\varepsilon$-approximate
$\mathrm{SU}(2)$-symmetric $2$-design after
$O(n^4+n^3\log\varepsilon^{-1})$ random gates.  For every fixed $\eta>0$,
the CP-order mixing time is
$\Theta(n^3\log\varepsilon^{-1})$ whenever
$\log\varepsilon^{-1}\geqslant(2\log2+\eta)n$; in particular, it is
$\Theta(n^4)$ for $\varepsilon=e^{-\alpha n}$ with fixed
$\alpha>2\log2$.
\end{corollary}

\begin{proof}
By Proposition~\ref{prop:path-complete-ground-space-equality}, every path
ground vector is completely symmetric and is annihilated by every reaction
projection.  Hence it is also a cycle ground vector.  Conversely,
$H_{C_n}\geqslant H_n$, so the two kernels agree and
Theorem~\ref{thm:main} gives the lower bound in
\eqref{eq:periodic-chain-gap}.
On each of the four tensor products in Proposition~\ref{prop:diffusive-upper-bound},
$H_{C_n}$ is unitarily equivalent to one half of the cycle Laplacian
restricted to $\mathbf 1^\perp$.  Its spectrum gives \eqref{eq:cycle-laplacian-branch} and the
upper bound in \eqref{eq:periodic-chain-gap}.

Since $\mathsf M_{C_n}=I-n^{-1}H_{C_n}$ is positive
semidefinite, the spectral theorem gives \eqref{eq:periodic-exact-mixing}.
Together with \eqref{eq:periodic-chain-gap}, this gives the operator-norm
mixing time.  For $n\geqslant5$, the fixed-space identification from
Corollary~\ref{cor:moment-mixing} and \cite[Lemma~B1, Eq.~(B10)]{LiEtAl} give
\eqref{eq:periodic-cp-comparison}.  The logarithmic comparison used in the
proof of Corollary~\ref{cor:sharp-cp-mixing} gives the remaining assertions.
\end{proof}

\subsection{All-to-all circuits}

The second-moment operator for a uniformly sampled edge of the complete
graph is $\mathsf M_{K_n}=I-\binom n2^{-1}H_{K_n}$.
Combining Theorem~\ref{thm:all-to-all}, the spectral theorem and
\cite[Lemma~B1]{LiEtAl}, we obtain the following result.

\begin{corollary}\label{cor:all-to-all-mixing}
Let $n\geqslant2$.  Then, for every integer $\ell\geqslant0$,
\begin{equation}\label{eq:all-to-all-mixing}
 \operatorname{gap}(\mathsf M_{K_n})
 \geqslant\frac{2c_0}{n(n-1)},
 \qquad
 \left\|\mathsf M_{K_n}^\ell
 -\mathsf P_n\right\|
 \leqslant\exp\!\left(-\frac{2c_0\ell}{n(n-1)}\right).
\end{equation}
Consequently, for $n\geqslant5$ and $0<\varepsilon<1$, an $\varepsilon$-approximate
$\mathrm{SU}(2)$-symmetric $2$-design is obtained after
$O(n^3+n^2\log\varepsilon^{-1})$ all-to-all random eSWAP gates.
\end{corollary}

Table~\ref{tab:geometry-scales} summarises the convergence bounds.
One brickwork period consists of an odd and an even layer, so its number of
periods and its parallel depth differ only by a factor of two.

\begin{table}[!ht]
\centering
\small
\renewcommand{\arraystretch}{1.16}
\begin{tabular}{@{}p{0.22\textwidth}p{0.25\textwidth}p{0.25\textwidth}p{0.19\textwidth}@{}}
\toprule
{\raggedright circuit\par}
& {\raggedright contraction rate\par}
& {\raggedright operator-norm mixing\par}
& {\raggedright CP-order mixing at $\varepsilon=e^{-\alpha n}$\par} \\
\midrule
open, random bond
& $1-q_n=\Theta(n^{-3})$
& $\Theta(1+n^3\log\varepsilon^{-1})$ gates
& $\Theta(n^4)$ gates \\
periodic, random bond
& $1-t_n=\Theta(n^{-3})$
& $\Theta(1+n^3\log\varepsilon^{-1})$ gates
& $\Theta(n^4)$ gates \\
open brickwork
& $1-b_n^2=\Theta(n^{-2})$
& $\Theta(1+n^2\log\varepsilon^{-1})$ periods
& $\Theta(n^3)$ periods \\
all-to-all, random edge
& $\operatorname{gap}(\mathsf M_{K_n})=\Omega(n^{-2})$
& $O(1+n^2\log\varepsilon^{-1})$ gates
& $O(n^3)$ gates \\
\bottomrule
\end{tabular}
\caption{Convergence bounds proved in this paper.  In the last column,
$\alpha>2\log2$ is fixed; the all-to-all upper bound holds for every fixed
$\alpha>0$.}
\label{tab:geometry-scales}
\end{table}

\section{Further questions}\label{sec:discussion}

Theorem~\ref{thm:main} determines the diffusive order of the open-chain gap.  It is natural to ask whether
$n^2\Delta(H_n)$ converges and which representation contains an eigenvector
for the smallest positive eigenvalue.  Determining the asymptotic size of
the unnormalised complete-graph gap would sharpen the all-to-all design bound.

The proof uses the two global replica pairings and the permutation module
on unordered pairs.  At higher moments, the fixed space involves more
pairing patterns; higher local dimension also introduces further Young
diagrams.  These extensions lead to two questions: which representations
occur in the off-diagonal block, and what gap bounds follow from deleting
one tensor factor?

\appendix

\section{Exact three-site calculations}\label{app:three-site}

We compute the complete-graph and path spectra using the same three-site
Schur--Weyl decomposition.  We first describe the common blocks, then
restrict to the diagonal $\Sn_3$-invariant subspace for the complete-graph
calculation.  The path calculation uses each block in full.

For one qubit replica,
\begin{equation}\label{eq:three-site-one-replica-schur-weyl}
 (\C^2)^{\otimes3}
 \cong
 \Symm^3(\C^2)\otimes\Specht{(3)}
 \oplus
 \Schur{(2,1)}\C^2\otimes\Specht{(2,1)}.
\end{equation}
The Schur factors have dimensions $4$ and $2$, respectively.  Let
$\mathcal V_3=\Specht{(2,1)}$ be the two-dimensional standard
representation.  In a Young orthogonal basis,
\begin{equation}\label{eq:three-site-standard-transpositions}
 \rho(\transp{1}{2})
 =\begin{pmatrix}1&0\\0&-1\end{pmatrix},
 \qquad
 \rho(\transp{2}{3})
 =\begin{pmatrix}
  -\frac12&\frac{\sqrt3}{2}\\[1mm]
  \frac{\sqrt3}{2}&\frac12
 \end{pmatrix}.
\end{equation}
The remaining transposition is represented by
\[
 \rho(\transp{1}{3})
 =\rho(\transp{1}{2})
  \rho(\transp{2}{3})
  \rho(\transp{1}{2}).
\]

Let $S\subseteq\{1,2,3,4\}$ be the set of replicas carrying the standard
Specht factor.  The other replicas carry
the trivial factor $\Specht{(3)}$.  The Schur factors in
\eqref{eq:three-site-one-replica-schur-weyl} determine the multiplicity of each tensor product of Specht modules and do
not affect the eigenvalues of the corresponding restriction.

For a transposition matrix $T$ in the standard representation, let $T_a$
act as $T$ on the factor corresponding to replica $a$ when $a\in S$, and
as the identity otherwise.  Write
$\mathsf P_{a,\pm}(T)=\tfrac12(I\pm T_a)$, and set
\begin{equation}\label{eq:active-replica-reaction-projection}
 q_S(T)
 =\mathsf P_{1,+}(T)\mathsf P_{2,+}(T)
  \mathsf P_{3,-}(T)\mathsf P_{4,-}(T)
 +\mathsf P_{1,-}(T)\mathsf P_{2,-}(T)
  \mathsf P_{3,+}(T)\mathsf P_{4,+}(T).
\end{equation}

\subsection{The three-particle complete-graph spectrum}\label{app:complete-three-particle}

Restricting to $\mathcal S_3$ selects the diagonal
invariant subspace
\[
 \mathcal I_S
 =\left(
   (\mathcal V_3)^{\otimes|S|}
  \right)^{\Sn_3}.
\]
On this subspace, the complete-graph reaction Hamiltonian is
\begin{equation}\label{eq:complete-three-active-block}
 R_{3,S}
 =q_S\!\left(
   \rho(\transp{1}{2})
  \right)
 +q_S\!\left(
   \rho(\transp{2}{3})
  \right)
 +q_S\!\left(
   \rho(\transp{1}{3})
  \right).
\end{equation}
The invariant dimensions are
\[
 \dim\mathcal I_S
 =\begin{cases}
  1,&|S|=0,2,3,\\
  3,&|S|=4,\\
  0,&|S|=1.
 \end{cases}
\]
These dimensions follow by averaging the character of
$(\mathcal V_3)^{\otimes|S|}$ over $\Sn_3$.

Let $\mathbf e_0,\mathbf e_1$ be the basis used in
\eqref{eq:three-site-standard-transpositions}.  For $|S|=2$, the invariant
line is spanned by
$\Omega=\mathbf e_0\otimes\mathbf e_0+\mathbf e_1\otimes\mathbf e_1$.
Substituting into \eqref{eq:active-replica-reaction-projection} and
\eqref{eq:complete-three-active-block} gives
\[
 R_{3,S}\Omega
 =\begin{cases}
  \dfrac32\Omega,
  &S=\{1,2\}\text{ or }S=\{3,4\},\\[2mm]
  0,
  &\text{one replica is a ket and the other a bra}.
 \end{cases}
\]

For $|S|=3$, the invariant line is spanned, up to permutation of the tensor
factors, by
\[
 \Theta
 =\mathbf e_{000}-\mathbf e_{011}
  -\mathbf e_{101}-\mathbf e_{110},
\]
where
$\mathbf e_{abc}=\mathbf e_a\otimes\mathbf e_b\otimes\mathbf e_c$.
For each of the four choices of $S$, substitution gives
$R_{3,S}\Theta=\tfrac34\Theta$.

For $S=\{1,2,3,4\}$, an invariant basis is
\begin{align*}
 \mathbf z_1
 &=\mathbf e_{0101}-\mathbf e_{0110}
  -\mathbf e_{1001}+\mathbf e_{1010},\\
 \mathbf z_2
 &=\mathbf e_{0011}-\mathbf e_{0110}
  -\mathbf e_{1001}+\mathbf e_{1100},\\
 \mathbf z_3
 &=\mathbf e_{0000}+\mathbf e_{0110}
  +\mathbf e_{1001}+\mathbf e_{1111}.
\end{align*}
Let $\mathsf Z$ be the linear map with columns
$\mathbf z_1,\mathbf z_2,\mathbf z_3$.  The Gram matrix and the matrix of the
compressed quadratic form are
\[
 \mathsf Z^*\mathsf Z
 =\begin{pmatrix}
  4&2&-2\\2&4&-2\\-2&-2&4
 \end{pmatrix},
 \qquad
 \mathsf Z^*R_{3,S}\mathsf Z
 =\begin{pmatrix}
  0&0&0\\0&6&0\\0&0&0
 \end{pmatrix}.
\]
Hence
\[
 \det\!\left(
  \mathsf Z^*R_{3,S}\mathsf Z
  -\lambda\mathsf Z^*\mathsf Z
 \right)
 =-8\lambda^2(4\lambda-9),
\]
so the eigenvalues on $\mathcal I_S$ are $0,0,9/4$.  The $|S|=0$ block
has eigenvalue $0$.

Table~\ref{tab:complete-three-site} includes the multiplicities contributed
by the Schur factors.

\begin{table}[ht]
\centering
\begin{tabular}{c c c c c}
\toprule
standard replicas & number of choices
& $\dim\mathcal I_S$ & Schur multiplicity & eigenvalues\\
\midrule
none & $1$ & $1$ & $4^4=256$ & $0$\\
$\{1,2\}$ or $\{3,4\}$ & $2$ & $1$ & $2^2 4^2=64$ & $3/2$\\
one ket and one bra & $4$ & $1$ & $2^2 4^2=64$ & $0$\\
three replicas & $4$ & $1$ & $2^3 4=32$ & $3/4$\\
all four replicas & $1$ & $3$ & $2^4=16$ & $0,0,9/4$\\
\bottomrule
\end{tabular}
\caption{Exact block decomposition of $R_3$ on $\Symm^3(\C^{16})$.}
\label{tab:complete-three-site}
\end{table}

Summing over the blocks, the eigenvalue $0$ has multiplicity
$256+4\cdot64+2\cdot16=544$.  The multiplicities of $3/4$, $3/2$, and $9/4$
are $128$, $128$, and $16$, respectively.  This proves
Proposition~\ref{prop:complete-graph-three-particle-spectrum}.

\subsection{The three-site path spectrum}\label{app:path-three-particle}

For the path, we retain the full tensor product of the Specht factors.
Given a set $S$ of replicas carrying the standard representation and a
transposition matrix $T$,
define the local Hamiltonian term
\begin{equation}\label{eq:active-replica-moment-projection}
 h_S(T)=\frac12\left(I-T_1T_2T_3T_4\right)+q_S(T),
\end{equation}
where $T_a=I$ for $a\notin S$.  The corresponding path block is
\begin{equation}\label{eq:path-three-active-block}
 H_{3,S}
 =h_S\!\left(
   \rho(\transp{1}{2})
  \right)
 +h_S\!\left(
   \rho(\transp{2}{3})
  \right)
 \quad\text{on }(\mathcal V_3)^{\otimes|S|}.
\end{equation}

All entries of \eqref{eq:path-three-active-block} lie in $\mathbb Q(\sqrt3)$.
Substituting the two matrices in
\eqref{eq:three-site-standard-transpositions} into
\eqref{eq:active-replica-moment-projection} and expanding
$\det(\lambda I-H_{3,S})$ gives the characteristic polynomials in
Table~\ref{tab:path-three-site}.  The symmetries obtained by permuting the two
ket replicas, permuting the two bra replicas, and simultaneously exchanging
ket and bra replicas reduce the calculation to the six cases in the table.

\begin{table}[ht]
\centering
\begin{tabular}{c c l}
\toprule
standard replicas & block dimension & characteristic polynomial\\
\midrule
none & $1$ & $\lambda$\\[1mm]
one replica & $2$ & $(\lambda-\tfrac12)(\lambda-\tfrac32)$\\[1mm]
two ket or two bra & $4$ & $(\lambda-\tfrac34)(\lambda-\tfrac54)(\lambda-2)^2$\\[1mm]
one ket and one bra & $4$ & $\lambda(\lambda-\tfrac12)(\lambda-\tfrac32)(\lambda-2)$\\[1mm]
three replicas & $8$ & $(\lambda-\tfrac38)(\lambda-\tfrac12)^2
 (\lambda-\tfrac32)^2(\lambda-\tfrac{13}{8})(\lambda-2)^2$\\[1mm]
all four replicas & $16$ & $\lambda^2(\lambda-\tfrac12)^3
 (\lambda-\tfrac78)(\lambda-\tfrac98)
 (\lambda-\tfrac32)^3(\lambda-2)^6$\\
\bottomrule
\end{tabular}
\caption{Characteristic polynomials of the restrictions $H_{3,S}$.}
\label{tab:path-three-site}
\end{table}

The Schur factors in
\eqref{eq:three-site-one-replica-schur-weyl} affect only multiplicities.
The least positive root is $3/8$, attained in every block with three
standard factors.  Therefore $\Delta(H_3)=\frac38$.


\newpage
\begin{thebibliography}{99}
\setlength{\itemsep}{0pt}
\setlength{\parsep}{0pt}

\bibitem{BrandaoHarrowHorodecki}
F. G. S. L. Brand\~ao, A.W. Harrow, and M. Horodecki,
\textit{Local random quantum circuits are approximate polynomial-designs},
Communications in Mathematical Physics 346 (2016), 397--434.

\bibitem{CaputoKac}
P. Caputo,
\textit{On the spectral gap of the Kac walk and other binary collision processes},
ALEA. Latin American Journal of Probability and Mathematical Statistics 4 (2008), 205--222.

\bibitem{CaputoLiggettRichthammer}
P. Caputo, T.M. Liggett, and T. Richthammer,
\textit{Proof of Aldous' spectral gap conjecture},
Journal of the American Mathematical Society 23 (2010), 831--851.

\bibitem{DankertEtAl}
C. Dankert, R. Cleve, J. Emerson, and E. Livine,
\textit{Exact and approximate unitary $2$-designs and their application to fidelity estimation},
Physical Review A 80 (2009), 012304.

\bibitem{DenerisEtAl}
A.E. Deneris, P. Bermejo, P. Braccia, L. Cincio, and M. Cerezo,
\textit{Exact spectral gaps of random one-dimensional quantum circuits},
Physical Review A 112 (2025), 062619.

\bibitem{FultonHarris}
W. Fulton and J. Harris,
\textit{Representation theory: a first course},
Graduate Texts in Mathematics 129, Springer, New York, 1991.

\bibitem{GossetMozgunov}
D. Gosset and E. Mozgunov,
\textit{Local gap threshold for frustration-free spin systems},
Journal of Mathematical Physics 57 (2016), 091901.

\bibitem{Haferkamp}
J. Haferkamp,
\textit{Random quantum circuits are approximate unitary $t$-designs in depth $O(nt^{5+o(1)})$},
Quantum 6 (2022), 795.

\bibitem{Halmos}
P.R. Halmos,
\textit{Two subspaces},
Transactions of the American Mathematical Society 144 (1969), 381--389.

\bibitem{HarrowLow}
A.W. Harrow and R.A. Low,
\textit{Random quantum circuits are approximate $2$-designs},
Communications in Mathematical Physics 291 (2009), 257--302.

\bibitem{HarrowMehraban}
A.W. Harrow and S. Mehraban,
\textit{Approximate unitary $t$-designs by short random quantum circuits using nearest-neighbor and long-range gates},
Communications in Mathematical Physics 401 (2023), 1531--1626.

\bibitem{HearthEtAl}
S.N. Hearth, M.O. Flynn, A. Chandran, and C. R. Laumann,
\textit{Unitary $k$-designs from random number-conserving quantum circuits},
Physical Review X 15 (2025), 021022.

\bibitem{HearthErratum}
S.N. Hearth, M.O. Flynn, A. Chandran, and C. R. Laumann,
\textit{Erratum: Unitary $k$-designs from random number-conserving quantum
circuits [Phys. Rev. X 15, 021022 (2025)]},
Physical Review X 16 (2026), 019901.

\bibitem{HunterJones}
N. Hunter-Jones,
\textit{Unitary designs from statistical mechanics in random quantum circuits},
Preprint, arXiv:1905.12053 (2019).

\bibitem{Knabe}
S. Knabe,
\textit{Energy gaps and elementary excitations for certain VBS-quantum antiferromagnets},
Journal of Statistical Physics 52 (1988), 627--638.

\bibitem{LiEtAl}
Z. Li, H. Zheng, J. Liu, L. Jiang, and Z.-W. Liu,
\textit{Designs from local random quantum circuits with $\mathrm{SU}(d)$ symmetry},
PRX Quantum 5 (2024), 040349.

\bibitem{LiEtAlApplications}
Z. Li, H. Zheng, Y. Wang, L. Jiang, Z.-W. Liu, and J. Liu,
\textit{$\mathrm{SU}(d)$-symmetric random unitaries: quantum scrambling, error
correction, and machine learning},
npj Quantum Information 11 (2025), 158.

\bibitem{LiangZhuISWAPMaximises}
Y. Liang and H. Zhu,
\textit{iSWAP maximises the second-moment spectral gap in random quantum circuits},
Preprint, arXiv:2607.29551 (2026).

\bibitem{LiangZhuISWAPMatching}
Y. Liang and H. Zhu,
\textit{Spectral gaps of ironed two-qubit gadgets matching the iSWAP gap},
Preprint, arXiv:2607.28521 (2026).

\bibitem{LiuHulseMarvian}
H. Liu, A. Hulse, and I. Marvian,
\textit{Unitary designs from random symmetric quantum circuits},
Preprint, arXiv:2408.14463 (2024).

\bibitem{Marvian}
I. Marvian,
\textit{Restrictions on realizable unitary operations imposed by symmetry and locality},
Nature Physics 18 (2022), 283--289.

\bibitem{MarvianLiuHulse}
I. Marvian, H. Liu, and A. Hulse,
\textit{Rotationally invariant circuits: universality with the exchange
interaction and two ancilla qubits},
Physical Review Letters 132 (2024), 130201.

\bibitem{MitsuhashiEtAl}
Y. Mitsuhashi, R. Suzuki, T. Soejima, and N. Yoshioka,
\textit{Unitary designs of symmetric local random circuits},
Physical Review Letters 134 (2025), 180404.

\bibitem{MoudgalyaMotrunich}
S. Moudgalya and O.I. Motrunich,
\textit{Symmetries as ground states of local superoperators: hydrodynamic implications},
PRX Quantum 5 (2024), 040330.

\bibitem{Nachtergaele}
B. Nachtergaele,
\textit{The spectral gap for some spin chains with discrete symmetry breaking},
Communications in Mathematical Physics 175 (1996), 565--606.

\bibitem{Sagan}
B.E. Sagan,
\textit{The symmetric group: representations, combinatorial algorithms, and symmetric functions},
second edition, Graduate Texts in Mathematics 203, Springer, New York, 2001.

\bibitem{Sasada}
M. Sasada,
\textit{On the spectral gap of the Kac walk and other binary collision processes on $d$-dimensional lattice},
in K. Iohara, S. Morier-Genoud and B. R\'emy (eds.), \textit{Symmetries, integrable systems and representations}, Springer Proceedings in Mathematics and Statistics 40, Springer, London, 2013, 543--560.

\end{thebibliography}
\end{document}